\documentclass[11pt]{article}
\usepackage[utf8]{inputenc}
\pdfoutput=1 % if your are submitting a pdflatex (i.e. if you have
\usepackage{jheppub} % for details on the use of the package, please
\usepackage{todonotes}
\usepackage{amsthm}
\usepackage{mathtools,centernot} % added to get access to \coloneqq

\usepackage{tabularx,multirow,booktabs,boldline}
\newcolumntype{L}[1]{>{\raggedright\arraybackslash}m{#1}}
\newcolumntype{C}[1]{>{\centering\arraybackslash}m{#1}}
\newcolumntype{R}[1]{>{\raggedleft\arraybackslash}m{#1}}

\newlength\myheight
\newlength\mydepth
\settototalheight\myheight{Xygp}
\newcommand*\inlinegraphics[1]{%
  \settototalheight\myheight{Xygp}%
  \settodepth\mydepth{Xygp}%
  \raisebox{-\mydepth}{\includegraphics[height=\myheight]{#1}}%
}

\usepackage{enumitem}

\usepackage{euscript,amsmath,amsthm,amsfonts,amssymb}%,show labels}
\usepackage{subcaption}
\hypersetup{%colorlinks,
	pdfencoding=auto
}

\newcommand{\R}{\mathbf{R}}

\newcommand{\M}{\mathcal{M}}
\newcommand{\N}{\mathcal{N}}
\newcommand{\I}{\mathcal{I}}
\newcommand{\PP}{\mathcal{P}}
\newcommand{\dM}{\partial\mathcal{M}}
\newcommand{\vv}{\vec v}
\newcommand{\GN}{G_{\rm N}}

\newcommand{\no}{\centernot\cap}

\DeclareMathOperator{\area}{area}

\DeclareMathOperator{\ceil}{ceil}

\makeatletter
\g@addto@macro\bfseries{\boldmath} % Makes math appearing in bolded text also bold
\makeatother

\newtheorem{theorem}{Theorem}[section]
\newtheorem{lemma}[theorem]{Lemma}
\newtheorem{corollary}[theorem]{Corollary}

\title{\boldmath Crossing versus locking: Bit threads and continuum multiflows}

\author{Matthew Headrick,}
\author{Jesse Held,}
\author{and Joel Herman}

\affiliation{Martin Fisher School of Physics \\ Brandeis University\\ Waltham MA 02453\\ USA}

\emailAdd{headrick@brandeis.edu}
\emailAdd{jmheld@brandeis.edu}
\emailAdd{joelaronherman@brandeis.edu}

\preprint{BRX-TH-6667}

\abstract{
Bit threads are curves in holographic spacetimes that manifest boundary entanglement, and are represented mathematically by continuum analogues of network flows or multiflows. Subject to a density bound, the maximum number of threads connecting a boundary region to its complement computes the Ryu-Takayanagi entropy. When considering several regions at the same time, for example in proving entropy inequalities, there are various inequivalent density bounds that can be imposed. We investigate for which choices of bound a given set of boundary regions can be ``locked'', in other words can have their entropies computed by a single thread configuration. We show that under the most stringent bound, which requires the threads to be locally parallel, non-crossing regions can in general be locked, but crossing regions cannot, where two regions are said to cross if they partially overlap and do not cover the entire boundary. We also show that, under a certain less stringent density bound, a crossing pair can be locked, and conjecture that any set of regions not containing a pairwise crossing triple can be locked, analogously to the situation for networks.
}

\begin{document} 
\maketitle
\flushbottom

\section{Introduction}

The celebrated Ryu-Takayanagi (RT) holographic entanglement entropy formula \cite{Ryu:2006bv} provides a simple and direct relationship between a quantum mechanical property of the boundary theory and a geometric property of the bulk spacetime. Specifically, it gives the entropy $S(A)$ of a boundary region $A$ as\footnote{In this paper we restrict ourselves to the domain of validity of the RT formula in the form \eqref{RT}, namely holographic field theories in a limit where the dual is governed by classical Einstein gravity, in a static or time-reflection-invariant state. Relaxing each of these limitations leads to very interesting questions about both entropy inequalities and bit threads; however, we will not venture to address those questions here.}
\begin{equation}\label{RT}
    S(A) = \frac{1}{4G_{\rm N}}\text{area}(m(A)),
\end{equation}
where $m(A)$ is the minimal-area bulk surface homologous to $A$. Among many other applications, the RT formula enables the study of entropy inequalities in the holographic setting by geometric means. These include inequalities that are obeyed by general quantum states, such as positivity, subadditivity, and strong subadditivity \cite{Headrick:2007km,Headrick:2013zda},\footnote{We use $A,B,\ldots$ to represent non-empty, non-overlapping boundary regions, which we refer as \emph{elementary} regions. More generally, a \emph{region} is an elementary region or union thereof (where $AB$ denotes $A\cup B$, etc.). We do not impose any restrictions on the dimension or topology of the bulk or on the topology of the boundary regions. Formal definitions will be given in subsection \ref{sec:manifold}.}
\begin{equation}\label{inequalities1}
S(A) \ge 0\,,\qquad
S(A)+S(B)\ge S(AB)\,,\qquad
S(AB)+S(BC)\ge S(B)+S(ABC)\,.
\end{equation}
They also include inequalities that are not obeyed by general quantum states, and therefore indicate some special properties of holographic states. The two simplest such inequalities are monogamy of mutual information (MMI) \cite{Hayden:2011ag,Headrick:2013zda},
\begin{equation}\label{MMI}
S(AB)+S(AC)+S(BC)\ge S(A)+S(B)+S(C)+S(ABC)
\end{equation}
and the following inequality for five regions, invariant under dihedral permutations \cite{Bao:2015bfa}:
\begin{multline}\label{dihedral}
S(ABC)+S(BCD)+S(CDE)+S(DEA)+S(EAB)\\
\ge S(AB)+S(BC)+S(CD)+S(DE)+S(EA)+S(ABCDE)\,.
\end{multline}
An infinite number of further inequalities, for 5 and more regions, has been found explicitly \cite{Bao:2015bfa}, and it is believed that there exist many more.
%Matt added footnote
\footnote{For further recent work on higher entropy inequalities, see \cite{Czech:2019lps,Brown:2015waa,Bao:2015boa,Bao:2017oms,Bao:2018wwd,Hubeny:2018trv,Hubeny:2018ijt,He:2019ttu,Caginalp:2019mgu,Cuenca:2019uzx,He:2020xuo}.}

All of the above inequalities are proved from the RT formula by a certain type of inclusion-exclusion argument \cite{Headrick:2007km}. This argument invokes the homology condition on the minimal surfaces, which says that, for a given boundary region $A$, there exists a bulk region $r(A)$ interpolating between $A$ and $m(A)$. The argument then considers the boundaries of regions obtained by taking unions, intersections, and differences of the regions associated to the terms on the left-hand side of each inequality. For example, for subadditivity, which has $S(A)$ and $S(B)$ on the left-hand side, one considers the union of their respective bulk regions:
\begin{equation}
\tilde r(AB) :=r(A)\cup r(B)\,.
\end{equation}
Its boundary $\tilde m(AB)$ has area no larger than $\area(m(A))+\area(m(B))$; on the other hand, being homologous to $AB$, it bounds the area of the minimal surface $m(AB)$ from above, giving
\begin{equation}
S(A)+S(B) = \frac1{4G_{\rm N}}\left(\area(m(A))+\area(m(B))\right) \ge \frac1{4G_{\rm N}}\area(\tilde m(AB)) \ge S(AB)\,.
\end{equation}
All known inequalities obeyed by RT entropies have been proved using essentially more elaborate versions of this argument. Unfortunately, although the RT formula makes finding and proving such inequalities relatively easy, it does not give much intuitive insight into their physical significance.

Some of the above inequalities have also been proved using the language of bit threads \cite{Freedman:2016zud,Cui:2018dyq}, which is in a certain sense dual to that of minimal surfaces. Bit threads are unoriented bulk curves ending on the boundary, subject to the rule that the density at any point in the bulk may not exceed $1/(4G_{\rm N})$. This bound in particular implies that the number of threads crossing the minimal surface $m(A)$ cannot exceed its area divided by $4G_{\rm N}$, hence the number $N_{A:\bar A}$ connecting $A$ and its complement does not exceed $S(A)$:
\begin{equation}\label{NboundA}
N_{A:\bar A}\le S(A)\,.
\end{equation}
Borrowing terminology from the theory of flows on networks, we say that a thread configuration \emph{locks} the region $A$ if the bound \eqref{NboundA} is saturated. In fact, the bound is tight: for any $A$, there exists a locking thread configuration:
\begin{equation}\label{maxN}
\max N_{A:\bar A} = S(A)
\end{equation}
(where the maximum is over allowed thread configurations). This theorem, which is the continuum analogue of the celebrated max flow-min cut theorem for graphs, is proved in three steps: (1) rewrite the problem of maximizing $N_{A:\bar A}$ as a convex program in terms of vector fields; (2) apply strong duality of convex programs to obtain a convex minimization program; (3) show that the latter reduces to finding the minimal surface homologous to $A$. (See \cite{Headrick:2017ucz} for a detailed exposition of the proof and further references.)

Eq.\ \eqref{maxN} leads to a proof of subaddditivity as follows. Consider a thread configuration that locks $AB$, $N_{AB:O}=S(AB)$, where $O:=\overline{AB}$. This configuration will contain some number $N_{A:O}$ of threads connecting $A$ to $O$, $N_{B:O}$ connecting $B$ to $O$, and $N_{A:B}$ connecting $A$ to $B$. (It may also contain threads connecting $A$ itself, etc.) Then,
\begin{equation}\label{SAproof}
S(A)+S(B) \ge N_{A:BO}+N_{B:AO}=N_{A:O}+N_{B:O}+2N_{A:B}\ge N_{A:O}+N_{B:O}=N_{AB:O}=S(AB)\,.
\end{equation}
The first inequality comes from \eqref{NboundA}, using the fact that $\bar A=BO$, etc.

A stronger version of the max flow-min cut theorem shows that nested regions (e.g.\ $A$ and $AB$) can be locked by a single thread configuration \cite{Freedman:2016zud,Headrick:2017ucz}. Applying this fact to $B$ and $ABC$, and using the same logic as in \eqref{SAproof}, leads immediately to the strong subadditivity inequality. Conceptual implications of these proofs for the encoding of boundary entanglement in bulk geometry were discussed in \cite{Freedman:2016zud}. Note that the two proof strategies are in some sense mirror images of each other: with minimal surfaces, we start with a set of surfaces that calculates the left-hand side and then use it to bound the right-hand side; with threads, we start with a  configuration that calculates the right-hand side and then use it to bound the left-hand side.

In \cite{Cui:2018dyq}, another locking theorem for multiple regions was proved: any set of non-overlapping regions can be locked. Applying this theorem to $A$, $B$, $C$, and $O:=\overline{ABC}$ leads immediately to the MMI inequality \eqref{MMI}. The existence of such locking configurations also motivated a proposed explanation for why holographic states obey MMI, called ``bipartite dominance''. Bipartite dominance is the idea that a pure holographic state with the boundary divided into three regions contains mostly bipartite entanglement, with only a small amount of tripartite entanglement (where ``small'' means contributing at order 1 to the entropy, versus at order $1/\GN$ for the bipartite entanglement). This statement implies the MMI inequality, and also matches  the calculable entanglement structure of stabilizer tensor network models of holography \cite{Nezami:2016zni}. On the other hand, an argument against bipartite dominance has recently been given \cite{Akers:2019gcv}.

The proof of the multi-region locking theorem in \cite{Cui:2018dyq} used the same convex duality technique described below \eqref{maxN} for the max flow-min cut theorem. While this proof is non-constructive, a method for constructing multi-region locking configurations in certain cases was described in \cite{Hubeny:2018bri} (see also \cite{Agon:2018lwq}). It was furthermore conjectured in  \cite{Hubeny:2018bri} that such a configuration can be constructed in which the threads connecting different pairs of boundary regions are segregated from each other, in other words occupy non-overlapping bulk regions.

In this paper we further explore locking thread configurations, and the circumstances under which they do or do not exist. One long-term goal of this investigation---which we do not achieve here---is to prove the higher entropy inequalities such as \eqref{dihedral} using bit threads, and thereby hopefully to derive some intuition for their underlying meaning. An important feature that distinguishes the ``lower'' inequalities \eqref{inequalities1} and \eqref{MMI} from the higher ones is that the composite regions on the right-hand sides of the higher ones ``cross''. Again we are following the terminology from network theory: two boundary regions are said to cross if they partially overlap and do not cover the whole boundary; for example, $AB$ crosses $BC$ (assuming $\overline{ABC}$ is non-empty), but $AB$ does not cross $A$, $C$, or $ABC$.

On a network, the analogue of a boundary region is a subset of the network terminals. There exists a well-developed theory, going under the name \emph{multiflow} (or \emph{multicommodity flow}) \emph{locking problems}, on the question of when a set of terminal subsets can be locked. For example, it is known that a set of terminal subsets can be locked if it do not contain a triple of subsets that cross pairwise \cite{KarzanovLomonosov}. On the other hand, simple counterexamples show that a set of terminal subsets containing a crossing triple cannot in general be locked. (See Appendix \ref{sec:networks} for a detailed discussion and further references.)

In this paper, working in the Riemannian manifold setting, we similarly relate the ability to lock a set of boundary regions to their crossing properties. However, we will see that the analogy between networks and manifolds is not straightforward; one cannot simply assume that theorems that hold on networks necessarily hold on manifolds. One significant complication arises from what one means precisely by ``density'' in the rule that the thread density must not exceed $1/(4\GN)$. Consider the threads passing through a small neighborhood. If these threads are parallel to each other, then there is a clear definition of their density, namely the number per unit transverse area. However, if they are not parallel, then there are various ways the density can be defined, and therefore bounded. In this paper we study three possible bounds: the  strongest sensible one, which we call $\nu_c$, \emph{requires} the threads to be locally parallel; the weakest sensible bound, $\nu_a$, bounds the number passing through any given surface; while an intermediate one, $\nu_v$, which is the easiest to analyze mathematically and was previously applied in \cite{Cui:2018dyq}, bounds the total length passing through a given volume. It is not clear physically which, if any, of these density bounds is the correct one for holographic bit threads. Indeed, finding the right density bound is one of the motivations for our investigation: If bit threads do play a fundamental role in connecting geometry to entanglement, then it is important to understand in detail the rules that they follow.

We first study the case of non-crossing boundary regions. In Theorem \ref{WCLT}, we show that under the $\nu_v$ density bound any crossing-free region set can be simultaneously locked. (This encompasses Theorems 1 and 2 of \cite{Cui:2018dyq} as well as the nesting property mentioned above.) This theorem is proven in two different ways: using convex dualization, and by decomposing the bulk into regions along minimal surfaces and then applying Theorem 1 of \cite{Cui:2018dyq} to each bulk region. Building on this theorem, we then prove in Theorem \ref{NOLT} that the threads in a locking configuration can be segregated, and therefore obey the stronger $\nu_c$ density bound. This proves the conjecture of \cite{Hubeny:2018bri} as a special case.

We then turn to the case of crossing region sets. We first show by an explicit counterexample that the $\nu_v$ density bound does \emph{not} allow crossing regions to be locked in general. However, we then show in Theorem \ref{continuumcrossing} that the weaker $\nu_a$ density bound \emph{does} allow any pair of crossing regions to be locked. We conjecture that the crossing-triple-free condition, which guarantees locking in the network setting, also does so on manifolds under the $\nu_a$ density bound. If true, this conjecture would restore the analogy between networks and Riemannian manifolds as far as lockability is concerned.

All of this leaves the problem of how to prove the higher entropy inequalities using bit threads. The regions on the right-hand side of the dihedral inequality \eqref{dihedral} are crossing triple-free. Thus, if our conjecture is correct, this inequality can be proved using bit threads subject to the $\nu_a$ density bound. On the other hand, interestingly, every other known holographic entropy inequality \emph{does} have a crossing triple among its right-hand side terms, and thus \emph{cannot} be proved this way. Proving these inequalities using threads thus seems to require a more radical relaxation in the density bound, such as having sets of threads living on different ``sheets'' and interacting only on the boundary.\footnote{Such a set-up was used in \cite{Harper:2019lff} for computing a quantity called holographic multipartite entanglement of purification.} We leave the analysis of such ideas to future work.

We are also left with the question of which, if any, of the density bounds we consider is the physically appropriate one for bit threads. The proof of the MMI inequality, as well as its conjectured interpretation in terms of the entanglement structure of holographic states, requires locking only a non-crossing set of regions. Our Theorem \ref{NOLT} shows that this can be accomplished even with the most stringent, $\nu_c$, density bound (which forces the threads to be locally parallel), suggesting that it is an appropriate condition to put on the threads. Weakening the density bound to the intermediate one $\nu_v$ does not buy us anything, in the sense of being able to lock more kinds of regions or prove more inequalities. As discussed above, relaxing further to the $\nu_a$ bound will at most allow us to prove one further inequality. A unified perspective on the higher inequalities in any case seems to require some other kind of structure. So an argument could be made that, overall, the results of this paper support the $\nu_c$ bound as the most physically well motivated, although clearly to resolve this question will require some new ideas.

The structure of this paper is as follows. In Section \ref{sec:background}, we establish the mathematical background for our work. In particular, we define the $\nu_{a,c,v}$ density bounds. We also define the notion of a \emph{multiflow} on a manifold, which is the main mathematical tool we use throughout the paper for proving theorems about threads, and study the precise relationship between thread configurations and multiflows. In Section \ref{section:WCLT}, we prove that an arbitrary crossing-free region set can be locked under the $\nu_v$ density bound. In Section \ref{section:NOLT}, we build on this result to prove that the more stringent $\nu_c$ bound also allows a crossing-free region set to be locked. In Section \ref{section:crossing}, we show that these results do not extend to crossing regions, by exhibiting a simple counterexample in which $AB$ and $BC$ cannot both be locked under the $\nu_v$ bound. Finally, in Section \ref{section:altBound}, we show that with the less stringent $\nu_a$ bound, any pair of crossing regions \emph{can} be locked, and we conjecture that this result extends to any crossing-triple-free region set, by analogy to the standard locking theorem on networks. Table \ref{summaryTable} summarizes the current state of knowledge, given the results in this paper, concerning locking of different types of region sets under different density bounds. The appendix contains a review of multiflows on networks and their relation to multiflows on manifolds, as well as the proof of a network theorem analogous to Theorem \ref{continuumcrossing}; this theorem is essentially a special case of the standard network locking theorem, but the proof method is novel as far as we know.

\section{Background}
\label{sec:background}

In this section we lay out the important concepts and definitions that we will use in the rest of the paper. These include the mathematical setting, the notions of threads, flows, and multiflows and their interrelations, important concepts such as crossing and locking, and the dualization of a certain class of multiflow maximization problems. The content of this section expands upon the framework laid out in the two papers \cite{Freedman:2016zud,Cui:2018dyq}. 

\subsection{Manifold and regions}
\label{sec:manifold}

The mathematical setting for our work, fixed throughout this paper, is a $d$-dimensional compact Riemannian manifold-with-boundary $\M$. $\M$ represents a time reflection-invariant Cauchy slice (moment of time symmetry) of a holographic spacetime. Its boundary $\dM$ represents a Cauchy slice of the conformal boundary where the field theory resides,\footnote{We are implicitly assuming that the dual state is pure, so that $\M$ is not bounded by any horizons; to our knowledge, any holographic mixed state (such as a thermal state) can be purified holographically, so this assumption comes without loss of generality. One can also consider ``internal'' boundaries that are not part of the conformal boundary where the field theory lives, such as confining walls and end-of-the-world branes. To avoid cluttering the discussion, we will not consider such internal boundaries here; however, they can easily be included simply by not allowing threads to end on them. See \cite{Headrick:2017ucz} for a discussion of how to include such boundaries in the dualization and max flow-min cut theorem.} with a suitable cutoff. A boundary (bulk) \emph{region} is a codimension-0 embedded submanifold-with-boundary of $\dM$ ($\M$). We abbreviate the flux of a vector field $\vv$ on a boundary region $A$ by $\int_A\vv$:
\begin{equation}
 \int_A\vv:=\int_A\sqrt h\,\hat n\cdot\vv\,,
\end{equation}
where $h$ is the determinant of the induced metric and $\hat n$ is the inward-directed unit normal vector on $\dM$.

The union of two boundary (bulk) regions $A,A'$ is a boundary (bulk) region. The same is not necessarily true of their intersection, which may contain codimension-1 or higher subsets. We say that $A$, $A'$ are \emph{non-overlapping}, and write
\begin{equation}\label{nodef}
A\no A'\,,
\end{equation}
if $A\cap A'$ is empty or has codimension-1 or higher, in other words if the interiors of $A$ and $A'$ do not intersect.

We will be considering various sets of (possibly overlapping) boundary regions. For this purpose, it is notationally convenient to effectively discretize $\dM$ by fixing a ``basis'' of non-overlapping  regions covering $\dM$, of which the various regions we consider are unions. Specifically, the \emph{elementary regions} $A_1,\ldots,A_n$ (which we sometimes write $A,B,\ldots$) are such that
\begin{equation}
A_i\no A_j\quad (i\neq j)\,,\qquad\bigcup_{i=1}^nA_i = \dM\,.
\end{equation}
The $A_i$ may have arbitrary topology; for example, we do not require them to be connected. Throughout this paper, the indices $i,j,k,l$ run from 1 to $n$.

A \emph{composite region} $A_I$ is a union of elementary regions:\footnote{This notation is borrowed from \cite{Bao:2015bfa}. However, note that, in contrast to that reference, the $A_i$ cover all of $\dM$; we do not include a ``purifying'' region $O$.}
\begin{equation}
A_I = \bigcup_{i\in I}A_i\,,
\end{equation}
%Joel%
%where $I\subseteq[n]:=\{1,\ldots,n\}$%
%Matt changed back to \subseteq
where $I\subseteq[n]:=\{1,\ldots,n\}$. (The elementary regions, along with the empty set and $\dM$, are thus special cases of composite regions.) We will also sometimes denote a composite region simply by writing its component elementary regions, for example $AB=A_{\{1,2\}}$. The complement of an index set is taken within $[n]$:
\begin{equation}
\bar I:=[n]\setminus I\,.
\end{equation}
As a subset of $\dM$, $A_{\bar I}$ is the closure of the setwise complement $\dM\setminus A_I$.

For convenience we assume the generic situation in which each composite region $A_I$ has a unique minimal-area homologous surface, which we denote $m(A_I)$, with $S(A_I)$ its area:
\begin{equation}
S(A_I):=\area(m(A_I))\,,
\end{equation}
In the holographic setting, in units where $4\GN=1$, $S(A_I)$ is the entropy of $A_I$ according to the Ryu-Takayanagi formula 
%Matt added citation
\cite{Ryu:2006bv}. We denote the homology region $r(A_I)$ (so $\partial r(A_I) = A_I\cup m(A_I)$). Under taking the complement, we have
\begin{equation}
m(A_{\bar I})=m(A_I)\,,\qquad r(A_{\bar I}) = \overline{r(A_I)}\,,\qquad S(A_{\bar I})=S(A_I)\,. 
\end{equation}
And of course we have
\begin{equation}
m(\emptyset)=m(\dM)=\emptyset\,,\qquad
r(\emptyset)=\emptyset\,,\qquad r(\dM)=\M\,,\qquad S(\emptyset)=S(\dM)=0\,.
\end{equation}

Two index sets $I,J$, or the corresponding regions $A_I$, $A_J$, are said to \emph{cross} if all of the following index sets are non-empty:
\begin{equation}\label{crossingdef}
I\cap J\,,\quad
I\cap \bar J\,,\quad
\bar I\cap J\,,\quad
\bar I\cap\bar J\,.
\end{equation}
For example, if $n=4$, then $AB$ crosses $BC$, but does not cross $A$, $ABC$, or $D$. Note that crossing is invariant under taking the complement of either or both regions: if $A_I$ and $A_J$ cross, then so do $A_{\bar I}$ and $A_J$, $A_I$ and $A_{\bar J}$, and $A_{\bar I}$ and $A_{\bar J}$. $A_I$ and $A_J$ do \emph{not} cross if and only if at least one of the following conditions holds:
\begin{equation}\label{notcrossing}
I\cap J = \emptyset\,,\quad
I\subseteq J\,,\quad
J\subseteq I\,,\quad
I\cup J = [n]\,,
\end{equation}
or in terms of the regions themselves,
\begin{equation}\label{notcrossing2}
A_I\no A_J \,,\quad
A_I\subseteq A_J\,,\quad
A_J\subseteq A_I\,,\quad
A_I\cup A_J = \dM\,,
\end{equation}
By the nesting property of homology regions \cite{Headrick:2013zda}, their homology regions $r(A_I)$, $r(A_J)$ then obey the corresponding property or properties:
\begin{equation}\label{notcrossingU}
r(A_I)\no r(A_J) \,,\quad
r(A_I)\subseteq r(A_J)\,,\quad
r(A_J)\subseteq r(A_I)\,,\quad
r(A_I)\cup r(A_J) = \M\,.
\end{equation}

\subsection{Threads and density bounds}
\label{subsection:threads}

The basic object in the bit-thread formulation of holographic entanglement \cite{Freedman:2016zud} is the \emph{thread}, an unoriented curve in $\M$ with no endpoints in the interior (so that it either ends on $\dM$ or is closed\footnote{We permit closed threads purely as a matter of mathematical convenience, allowing us to pass more easily between the thread and flow languages. (Indeed, much of the power of the flow language derives from the fact that it involves only local constraints, using which one cannot rule out closed flow lines.) In the holographic setting, they do not (as far as we know) correspond to any form of entanglement. Mathematically, for the purposes of locking, our main interest in this paper, they play no role and can always be deleted.}). We will be considering collections of threads, on which we impose the requirement that the density is everywhere less than or equal to $1$.\footnote{Restoring units, the density is required to be less than or equal to $1/4\GN$. Since the threads are 1-dimensional, the density has units of length$^{1-d}$.} This raises the question of what we mean by the ``density'' of a set of threads. In fact, we will consider several possible notions of density in this paper. Throughout our discussion, we are not interested in individual threads, but only in a coarse-grained description of macroscopic numbers of threads on length scales much larger than the Planck scale (which we've set to 1 by setting $4G_{\rm N}=1$).\footnote{As discussed in \cite{Cui:2018dyq}, we can define a thread configuration mathematically as a set of unoriented curves equipped with a measure, allowing us to ``count'' the number passing through a surface or connecting one boundary region to another. However, in this paper we will not aim for mathematical precision in our discussion of threads.}

If the threads are locally parallel---i.e.\ parallel in every small neighborhood---then the only natural definition of density is the transverse density, which we call $\nu_t$; this is the number per unit area intersecting a small disk orthogonal to the threads. (Here ``small'' means small compared to the curvature scales of $\M$ but large compared to the Planck scale.) When the threads are not locally parallel, however, there are at least two natural notions of density, based on volume and area respectively. We define $\nu_v$ as the total length of threads contained in a small ball divided by its volume.  This notion of density was used in \cite{Cui:2018dyq}. We define $\nu_a$ as the number per unit area intersecting a small disk, maximized over the orientation of the disk. It is straightforward to see that $\nu_v\ge\nu_a$, and if the threads are locally parallel then $\nu_v=\nu_a=\nu_t$. We thus have two possible density bounds:
\begin{equation}\label{density1}
\nu_a\le1\,,
\end{equation}
and
\begin{equation}\label{density0}
\nu_v\le1
\end{equation}
with \eqref{density0} stronger than \eqref{density1}. A third possibility, stronger than both, is to \emph{require} the threads to be locally parallel everywhere:
\begin{equation}\label{densityc}
\text{threads locally parallel,}\quad\nu_t\le1\,;
\end{equation}
we call this the $\nu_c$ condition. One could easily invent yet other density conditions, but these three are the ones we will consider in this paper.

Any of the above conditions implies that the number of threads crossing any surface is bounded by its area. In particular, the number $N_{A_I:A_{\bar I}}$ connecting a boundary region $A_I$ and its complement is bounded by $S(A_I)$:
\begin{equation}\label{Nbound}
N_{A_I:A_{\bar I}}\le S(A_I)\,.
\end{equation}
%Matt: rearranged sentence to eliminate confusion
Furthermore, as we will show below, the max flow-min cut theorem applies for any of the density bounds; this theorem says
\begin{equation}\label{mfmc}
\max N_{A_I:A_{\bar I}} = S(A_I)\,,
\end{equation}
where the maximum is over thread configurations. %
%Joel%
%%
%and %
%applies for any of the density bounds. 
Thus the threads reproduce the RT formula for any of the density bounds. We say that a thread configuration that achieves the maximum \eqref{mfmc} \emph{locks} $A_I$. Clearly, a configuration that locks $A_I$ also locks $A_{\bar I}$. 
%Matt:
Under the bound \eqref{density0} or \eqref{densityc}, %
a necessary and sufficient condition for a thread configuration to lock $A_I$ is for the threads that intersect $m(A_I)$ to (1) be orthogonal to it; (2) saturate the density bounds \eqref{density0} and \eqref{densityc} on $m(A_I)$ (these are equivalent since the threads are locally parallel there); and (3) cross $m(A_I)$ only once.

The conditions \eqref{density1}, \eqref{density0}, \eqref{densityc} represent the full range, from weakest to strongest, of local density conditions that guarantee \eqref{mfmc}; this property is essential since it is what allows thread configurations to compute minimal surface areas and thereby RT entropies. \eqref{density1} simply states that the number of threads crossing a surface is bounded by its area, and is therefore the weakest condition guaranteeing the bound \eqref{Nbound}. Unfortunately, \eqref{density1} is hard to work with mathematically, for example to prove locking theorems. (A locking theorem is an existence theorem for thread configurations that lock sets of composite regions.) We will discuss it further and conjecture a locking theorem in section \ref{section:altBound}. \eqref{density0} is stronger but much easier to work with. It was used in \cite{Cui:2018dyq} to prove two locking theorems. We will generalize these theorems in section \ref{section:WCLT} to a general locking theorem for sets of regions with no crossing. However, in section \ref{section:crossing} we will show that the bound can forbid locking of crossing regions. \eqref{density1} and \eqref{density0} are both natural generalizations of the notion of a capacity constraint for a multicommodity flow on a graph, which is a bound on the total amount of all the commodities flowing through a given edge; the distinction between \eqref{density1} and \eqref{density0} only occurs when threads cross at an angle, but there is no analogue of this phenomenon on a graph. \eqref{densityc} is the strongest local condition that can be imposed while preserving \eqref{mfmc}; again it is more difficult to work with than \eqref{density0}, in particular because, when we represent thread configurations by vector fields, it does not correspond to a convex condition. Nonetheless, in section \ref{section:NOLT}, we will succeed in proving a locking theorem for non-crossing regions under the condition
%Matt:
%
%Joel%
%for%
%using%
 \eqref{densityc}.

\subsection{Flows}
\label{subsection:Flows_Basics}

We are interested in studying the behavior of threads on macroscopic scales---much larger than the Planck scale. A mathematically convenient continuum description is provided by vector fields. Specifically, a divergenceless vector field $\vv$ represents a set of threads with transverse density $\nu_t=|\vv|$ and, where $\vv\neq0$, direction parallel to $\vv$. Therefore the thread configuration, which is automatically locally parallel, obeys all of the density bounds discussed in the previous section if and only if $|\vv|\le1$ everywhere. We call a vector field that obeys $\nabla\cdot\vv=0$ and $|\vv|\le1$ a \emph{flow}. We note that the density of thread endpoints at a point on $\dM$ is $|\hat n\cdot\vv|$.

While a flow gives rise to a thread configuration, the mapping from flows to thread configurations is neither one-to-one nor onto. It is not one-to-one because the vector field implicitly assigns an orientation to the threads. Changing this orientation on any subset of the threads leads to a different flow; for example, $\vv$ and $-\vv$ correspond to the same thread configuration. The mapping is not onto because a configuration in which the threads are not locally parallel, i.e.\ in which threads cross each other at angles, is not represented by any flow.

In general for any flow $\vv$ and region $A_I$ we have
\begin{equation}\label{flowbound}
\int_{A_I}\vv\le S(A_I)\,.
\end{equation}
The max flow-min cut theorem (see \cite{Headrick:2017ucz} and references therein) says that the bound \eqref{flowbound} is tight:
\begin{theorem}[Riemannian max flow-min cut]\label{mfmctheorem}
Let $A_I$ be a boundary region. Then there exists a flow $\vv$ such that
\begin{equation}\label{mfmc2}
\int_{A_I}\vv = S(A_I)\,.
\end{equation}
\end{theorem}
\noindent A flow $\vv$ is said to \emph{lock} $A_I$ if the underlying thread configuration does so. \eqref{mfmc2} provides a sufficient condition. (It is not a necessary condition. For example, if $S(A_I)>0$ and $\vv$ achieves the maximum then $-\vv$ does not, although they represent the same thread configuration.)

\subsection{Multiflows}
\label{subsection:Multiflows_Basics}

\subsubsection{Definition}

Throughout this paper, we will be very interested in how many threads connect various boundary regions in some given configuration. To keep track of this data, it is useful to define a divergenceless vector field $\vv_{ij}$ for each pair of distinct elementary regions, representing the threads connecting  $A_i$ and $A_j$; we thus require $\hat n\cdot\vv_{ij}=0$ on $A_k$ for $k\neq i,j$. We don't need an independent vector field $\vv_{ji}$; we therefore define only $\vv_{ij}$ with $i<j$. We will denote by $V$ the set of vector fields $(\vv_{ij})_{i<j}$. We will call each $\vv_{ij}$ a \emph{component flow} and the corresponding set of threads a \emph{thread bundle}. (In addition to threads connecting $A_i$ and $A_j$, this bundle may include closed threads and ones that connect $A_i$ to itself or $A_j$ to itself.)

Each thread bundle is locally parallel and has density $|\vv_{ij}|$. However, at a given point there may be multiple non-zero component flows, and these may not be parallel. At such a point threads cross at angles, leading to different values for the volume and area densities defined in subsection \ref{subsection:threads}. The total volume density is simply the sum of the densities of the individual bundles:
\begin{equation}\label{nu0def}
\nu_v(V)=\sum_{i<j}|\vv_{ij}|\,.
\end{equation}
The density on a small disk normal to the unit vector $\hat n$ for $\vv_{ij}$ is $|\hat n\cdot\vv_{ij}|$, so the total area density is
\begin{equation}\label{nu1def}
\nu_a(V)=\max_{\hat n}\sum_{i<j}|\hat n\cdot\vv_{ij}|\,.
\end{equation}
To implement the condition \eqref{densityc}, which forbids bundles from overlapping, we require $\nu_c(V)\le1$, where
\begin{equation}\label{nucdef}
\nu_c(V) := \sum_{i<j}\ceil(|\vv_{ij}|)\,,
\end{equation}
and $\ceil$ is the ceiling function (the smallest integer greater than or equal to its argument).\footnote{At first sight, the condition $\nu_c(V)\le1$ may appear stronger than \eqref{densityc}, since the former forbids bundles from overlapping, whereas the latter merely forces them to be parallel. However, since distinct bundles start and/or end on distinct boundary regions, in order to overlap they must be non-parallel somewhere.} These three collective norms\footnote{We are slightly abusing the term ``norm'', since $\nu_c$ does not obey the homogeneity property usually required of a norm.} are related as follows:
\begin{equation}
\nu_c(V) \ge\nu_v(V)\ge\nu_a(V)\,;
\end{equation}
therefore requiring $\nu_c(V)\le1$ imposes the strongest condition and $\nu_a(V)\le1$ the weakest.

The above conditions are summarized in the following definition. A \emph{$\nu_{v,a,c}$-multiflow} $V$ is a set of vector fields $\vv_{ij}$ ($i<j$) obeying:
\begin{align}
\left.\hat n\cdot\vv_{ij}\right|_{A_k}&=0\quad (k\neq i,j) \label{eq:flow_terminal} \\
\nabla\cdot\vv_{ij} &= 0 \label{eq:divergenceless1}\\
\nu_{v,a,c}(V) &\le 1 \,.\label{nubound}
\end{align}
The term ``multiflow'' as defined in \cite{Cui:2018dyq} would refer to what we would call a ``$\nu_v$-multiflow'' here.\footnote{There is also the trivial notational difference that in \cite{Cui:2018dyq} the fields $\vv_{ij}$ were also defined for $i\ge j$, but with the constraint $\vv_{ji}=-\vv_{ij}$.} The conditions \eqref{eq:flow_terminal} and \eqref{eq:divergenceless1} imply
\begin{equation}\label{vji}
\int_{A_j}\vv_{ij} = -\int_{A_i}\vv_{ij}\,.
\end{equation}

Note that \eqref{eq:flow_terminal}, \eqref{eq:divergenceless1} are linear equality constraints, and the functions $\nu_v$ and $\nu_a$ are convex functions; therefore the $\nu_v$- and $\nu_a$-multiflows form convex sets. On the other hand the condition $\nu_c(V)\le1$, which enforces non-overlap of thread bundles, is not convex.

\subsubsection{Relation to threads and flows}
\label{sec:multiflow relations}

Just as for flows, while a multiflow gives rise to a thread configuration obeying the corresponding density bound, the mapping is neither one-to-one nor onto. The mapping is not onto because, in the configuration corresponding to a multiflow, each thread bundle must be locally parallel. The mapping is not one-to-one because, again, we can reverse the direction of a component flow, or part thereof, without changing the corresponding thread configuration.\footnote{We can fix the ambiguity in the direction of the component flows by forcing the flow lines of $\vv_{ij}$ to go from $A_i$ to $A_j$. This would be accomplished by adding to the definition of a multiflow the further constraint $\hat n\cdot\vv_{ij}\ge0$ on $A_i$. This essentially amounts to fixing a gauge ambiguity and would not change anything in the main substance of the paper. It would lead to the following changes to the material in the rest of this section: It would guarantee that \eqref{fluxnumb} is obeyed; eliminate the possibility of threads connecting $A_i$ to itself, requiring us to delete such threads in the procedure below \eqref{fluxnumb} to convert a flow to a multiflow; make \eqref{mfmc2}, as applied to \eqref{v from vij}, a necessary (in addition to sufficient) condition for the multiflow $V$ to lock $A_I$; and change the constraints \eqref{eq:21} for the dual program to $\psi_{ij}\ge N_{ij}$ on $A_i$, $\psi_{ij}\le0$ on $A_j$. It would leave unchanged the dual program in the second form \eqref{dual2}.} If $\hat n\cdot\vv_{ij}\ge0$ on $A_i$ then
\begin{equation}\label{fluxnumb}
\int_{A_i}\vv_{ij} =N_{A_i:A_j}\,,
\end{equation}
where $N_{A_i:A_j}$ is the number of threads connecting $A_i$ and $A_j$ in the corresponding bundle.

At several points in this paper, it will be useful convert a flow into a multiflow or vice versa. To convert a flow into a multiflow, simply first convert it into a thread configuration and then separate the threads into bundles according to the regions they start and end on. Closed threads can be assigned to any bundle, and threads that start and end on the same region $A_j$ assigned to the $ij$ bundle for any $i<j$ or the $jk$ bundle for any $k>j$. The thread configuration is necessarily locally parallel, so each bundle will also be locally parallel, and therefore can be converted into a component flow. 
%Matt: removed the following confusing phrase and parenthetical comment
%$\vv_{ij}$ or $\vv_{jk}$. %
%Joel%
%(Its orientation can be chosen to obey the condition $\hat n\cdot\vv_{ij}\ge0$ on $A_i$ or $\hat n\cdot\vv_{jk}\ge0$ on $A_j$ and therefore %
%(Its orientation can be arbitrarily chosen to obey the condition $\hat n\cdot\vv_{ij}\ge0$ on $A_i$ and therefore %
%\eqref{fluxnumb}.) 
Since the thread bundles are non-overlapping, with density less than or equal to 1, the multiflow obeys the strictest norm bound, $\nu_c\le1$, hence either of the others as well.

Conversely, for any $\nu$, a $\nu$-multiflow $V$ can be converted to a flow by taking linear combinations of the component flows; the vector field
\begin{equation}\label{lincomb}
\vv = \sum_{i<j}\xi_{ij}\vv_{ij}\,,
\end{equation}
where each $\xi_{ij}$ is a constant between $-1$ and 1, is divergenceless and satisfies $|\vv|\le1$. Note that, while the decomposition of a flow into a multiflow described in the previous paragraph does not modify the underlying thread configuration, the linear combinations \eqref{lincomb} may do so.

\subsubsection{Locking}
\label{sec:locking}

Fix a region $A_I$ and multiflow $V$. We can pick out the components that flow from $A_I$ to $A_{\bar I}$ by setting
\begin{equation}
\xi_{ij} = \begin{cases}
+1\,,&\quad i\in I,j\in\bar I \\
-1\,,&\quad i\in\bar I,j\in I \\
0\,,&\quad \text{otherwise}
\end{cases}
\end{equation}
in \eqref{lincomb}, yielding the flow
\begin{equation}\label{v from vij}
\vv_I := \sum_{I\ni i<j\in\bar I}\vv_{ij} - \sum_{\bar I\ni i<j\in I}\vv_{ij}\,.
\end{equation}
Its flux is
\begin{equation}\label{fluxvI}
\int_{A_I}\vv_I=
\sum_{I\ni i<j\in\bar I}\int_{A_i}\vv_{ij}+
\sum_{\bar I\ni i<j\in I}\int_{A_i}\vv_{ij}
\end{equation}
(where we used \eqref{vji}). Being a flow, $\vv_I$ obeys \eqref{flowbound}. The multiflow is said to \emph{lock} $A_I$ if the underlying thread configuration does, and a sufficient condition is \eqref{mfmc2}. For any $I$ and any choice of $\nu$ there exists a locking $\nu$-multiflow, since we can convert a locking flow $\vv$ into a $\nu$-multiflow by the procedure described in the previous subsection.

A theorem of \cite{Cui:2018dyq} shows that there exists a $\nu_v$-multiflow that simultaneously locks all the elementary regions $A_i$; by a redefinition of the elementary regions, this also implies that any disjoint set of composite regions can be simultaneously locked. A principal goal of this paper is to investigate possible generalizations of this theorem. We say that a thread configuration or multiflow \emph{locks} a family $\I\subseteq2^{[n]}$ of subsets of $[n]$ if it locks $A_I$ for all $I\in\I$.

One of our main strategies for proving the possibility or impossibility of locking a given set $\I$ with a given type of multiflow is as follows. Given a multiflow $V$, we have, for all $I\in\I$,
\begin{equation}\label{fluxbound}
\int_{A_I}\vv_I\le S(A_I)\,.
\end{equation}
Therefore
\begin{equation}\label{total flux bound}
\sum_{I\in\I} \int_{A_I}\vv_I\le \sum_{I\in\I} S(A_I)\,,
\end{equation}
and furthermore \eqref{total flux bound} is saturated if and only if \eqref{fluxbound} is saturated for all $I\in\I$, which is true if and only if $\I$ is locked. So the strategy is to maximize the left-hand side of \eqref{total flux bound} over $\nu$-multiflows and seeing if it is saturated.

\subsection{Dualization of $\nu_v$-multiflow problems}
\label{subsection:dual}

We call the problem of maximizing the left-hand side of \eqref{total flux bound} (for a given $\nu$ and $\I$) the \emph{primal} program:
\begin{equation}\label{primal}
\text{Maximize }\sum_{I\in\I}\int_{A_I}\vv_I
\quad\text{over $\nu$-multiflows $V$}\,.
\end{equation}
As noted below \eqref{nubound}, the set of $\nu_v$- and $\nu_a$-multiflows form convex sets, so the primal program is a convex optimization problem for those two norm bounds. This allows us to bring to bear the power of the theory of convex optimization. In particular, it is often very useful to study the dual of a convex program, which under certain mild conditions has the same optimal value as the primal (so-called \emph{strong duality}). See \cite{Headrick:2017ucz} for a physicist-friendly introduction to strong duality and its application to Riemannian flow problems.

In this subsection we will derive the dual of the $\nu_v$ primal program. This is very similar to the dualizations considered in \cite{Cui:2018dyq}, so we will be brief. Unfortunately, the $\nu_a$ norm bound is much more difficult to work with and we were not able to derive the dual of the $\nu_a$ primal program in any useful form.

We define $n_{ij}$ to be the number of sets $I\in\I$ which contain $i$ and not $j$,
\begin{equation}\label{nijdef}
n_{ij}:=\left|\{I\in\I|i\in I,j\in \bar I\}\right|\,,
\end{equation}
and $N_{ij}$ as the symmetrized number:
\begin{equation}\label{Nijdef}
N_{ij}:=n_{ij}+n_{ji}\,.
\end{equation}
After applying \eqref{fluxvI} to the dual objective, this allows us to switch the order of the sums:
\begin{equation}
\sum_{I\in\I}\int_{A_I}\vv_I = 
\sum_{I\in\I}\sum_{I\ni i<j\in\bar I}\int_{A_i}\vv_{ij} +\sum_{I\in\I}\sum_{\bar I\ni i<j\in I}\int_{A_i}\vv_{ij} =
\sum_{i<j}N_{ij}\int_{A_i}\vv_{ij}\,.
\end{equation}

We will let the boundary condition \eqref{eq:flow_terminal} be implicit and the divergenceless constraint \eqref{eq:divergenceless1} and norm bound \eqref{nubound} be explicit. This means that we introduce Lagrange multipliers, $\psi_{ij}$ ($i<j$) and $\lambda$ respectively, for the latter two constraints. We then solve the maximization problem without imposing those constraints, but while still imposing \eqref{eq:flow_terminal}. Adding the Lagrange multiplier terms to the objective, we arrive at the following Lagrangian functional:
\begin{equation}
 L=  \sum_{i<j}N_{ij}\int_{A_{i}}\vv_{ij} +\int_{\M}\sqrt g\left(\lambda +\sum_{i<j}\left(\psi_{ij}\nabla\cdot \vv_{ij}-\lambda|\vv_{ij}|\right)\right).
\end{equation}
Integrating the divergence term by parts yields
\begin{equation}
L=  \sum_{i<j}\left(\int_{A_{i}}\vv_{ij}(N_{ij}-\psi_{ij})-\int_{A_j}\vv_{ij}\,\psi_{ij}\right) +\int_{\M}\sqrt g\left(\lambda -\sum_{i<j}\left(\nabla\psi_{ij}\cdot \vv_{ij}+\lambda|\vv_{ij}|\right)\right).
\end{equation}
Demanding that the maximum of the Lagrangian with respect to $\vv_{ij}$ be finite implies the following constraints for the dual program:
\begin{equation}
    \lambda\geq |\nabla\psi_{ij}|,
    \label{eq:20}
\end{equation}
\begin{equation}
    \psi_{ij}|_{A_{i}}=N_{ij} \,,\qquad\psi_{ij}|_{A_{j}}=0\,.
    \label{eq:21}
\end{equation}
Maximizing over $\vv_{ij}$ then yields just $\int_{\M}\sqrt g\lambda$ for the dual objective. The \emph{dual program} is thus
\begin{equation}
\text{Minimize $\int_{\M}\sqrt g\,\lambda$ with respect to $\lambda,\{\psi_{ij}\}$, subject to \eqref{eq:20}, \eqref{eq:21}.}
\end{equation} 
It is easily seen that Slater's condition is obeyed, so strong duality holds: the primal and dual programs have the same optimal values.

We can rewrite the dual problem in a way that removes the $\psi_{ij}$s entirely. Together, \eqref{eq:20} and \eqref{eq:21} are equivalent to the following two constraints on $\lambda$:
\begin{equation}\label{positive}
\lambda\ge0
\end{equation}
\begin{equation}
    \int_{c}ds\,\lambda \geq N_{ij}\quad\text{$\forall$ curve $c$ from $A_{i}$ to $A_{j}$}\,.
    \label{eq:pathBound}
\end{equation}
The fact that \eqref{eq:20} and \eqref{eq:21} imply \eqref{positive} and \eqref{eq:pathBound} is fairly straightforward. Showing the converse requires constructing the function $\psi_{ij}$ given a function $\lambda$ obeying \eqref{positive} and \eqref{eq:pathBound}. One way to do this is to define the function
\begin{equation}\label{phiidef}
\phi_i(x) := \inf\int_{A_i}^xds\,\lambda\,,
\end{equation}
where the infimum is over paths from $A_i$ to $x\in\M$, and then set
\begin{equation}
\psi_{ij}(x) = \min\{0,N_{ij}-\phi_i(x)\}\,.
\end{equation}
Therefore the dual program can be written
\begin{equation}\label{dual2}
 \text{Minimize $\int_{\M}\sqrt g\,\lambda$ with respect to $\lambda$, subject to \eqref{positive}, \eqref{eq:pathBound}}\,.\\
 \end{equation}

In general in this paper we will be fairly sloppy about the real-analysis aspects of the statements and proofs we make. However, one technical comment concerning the dual program may help to clear up some confusion on the reader's part. Below we will prove several bounds on the dual objective, such as \eqref{threadWD}, \eqref{dualbound}, \eqref{boundedbelow}, under various assumptions. In each case we will assume implicitly in the proof that $\lambda$ is a continuous function. On the other hand, the solution (i.e.\ optimal configuration) of \eqref{dual2} in some cases involves delta functions, which obviously are not continuous functions. The bounds are nontheless valid even for a configuration containing delta functions, including the solution, since a delta function can be obtained as a limit of a sequence of continuous functions with both the constraints  $\int_c ds\lambda$ and the objective $\int_{\M}\sqrt g\lambda$ converging to the appropriate values.

\subsubsection{Weak duality for threads}

The dual program \eqref{dual2} may seem somewhat distant from our original objects of interest, the bit threads, since we got to it by several steps (multiflows, convex programming, dualization). However, as we will now show, \eqref{dual2} can be directly related back to threads by taking the curves in the constraint \eqref{eq:pathBound} to be the threads themselves:

\begin{theorem}
Let $\I$ be a set of subsets of $[n]$ and $\lambda:\M\to\R$ a function obeying \eqref{positive}, \eqref{eq:pathBound}. Then for any thread configuration obeying the density bound $\nu_v\le1$ everywhere,
\begin{equation}\label{threadWD}
\int_{\M}\sqrt g\,\lambda \ge \sum_{I\in\I} N_{A_I:A_{\bar I}}\,.
\end{equation}
\end{theorem}

\begin{proof}
Divide $\M$ into cells $\mathcal N$, each of which is small enough that $\lambda$ can be considered constant and the thread configuration uniform within it. The density $\nu_v$ in a cell is the total length of threads within it divided by its volume $V_{\mathcal N}$:
\begin{equation}
\nu_v := \frac1{V_{\mathcal{N}}}\sum_c\int_{c\cap\mathcal{N}}ds\,,
\end{equation}
where the sum is over all the threads in the configuration (most of which will typically have $c\cap\mathcal{N}=\emptyset$). Since $\nu_v\le1$, and using \eqref{positive}, we have
\begin{equation}
V_{\mathcal{N}}\lambda(\mathcal{N})\ge \sum_c\int_{c\cap\mathcal{N}}ds\,\lambda(\mathcal{N})\,.
\end{equation}
Summing over the cells yields
\begin{equation}\label{threadbound1}
\int_{\M}\sqrt g\,\lambda\ge\sum_c\int_cds\,\lambda\,.
\end{equation}
Meanwhile,
\begin{equation}\label{threadbound2}
\sum_c\int_c ds\,\lambda\ge\sum_{i<j}\sum_{c:A_i-A_j}\int_cds\,\lambda\ge \sum_{i<j}N_{A_i:A_j}N_{ij} = \sum_{I\in\I} N_{A_I:A_{\bar I}}\,,
\end{equation}
where in the first inequality we separated the threads according to their endpoints, dropping any not connecting distinct elementary regions; $c:A_i-A_j$ means threads connecting $A_i$ and $A_j$; in the second inequality we used \eqref{eq:pathBound}; and in the equality we rearranged the sum using the definition of $N_{ij}$. Combining \eqref{threadbound1} and \eqref{threadbound2} yields \eqref{threadWD}.
\end{proof}

Thus any feasible function $\lambda$ imposes an upper bound on the total number of threads connecting the different regions and their complements. This is essentially a form of weak duality, but for thread configurations rather than flows. In particular, if it happens that
\begin{equation}\label{lambdaint}
\int_{\M}\sqrt g\,\lambda<\sum_{I\in\I} S(A_I)\,,
\end{equation}
then $\I$ \emph{cannot} be locked by a thread configuration obeying the $\nu_v$ density bound.

Of course, we already knew, by the duality between \eqref{primal} and \eqref{dual2}, that \eqref{lambdaint} implies the impossibility of a $\nu_v$-multiflow saturating \eqref{total flux bound}. However, since not every thread configuration can be written in terms of a multiflow, this leaves open the loophole that some thread configuration might be able to do it. \eqref{threadWD} closes this loophole. This will be important for us in section \ref{section:crossing}.

\section{\boldmath Non-crossing regions: \texorpdfstring{$\nu_v$}{𝜈_v} locking theorem}
\label{section:WCLT}

The following locking theorem was proven in \cite{Cui:2018dyq}:
\begin{theorem}[\cite{Cui:2018dyq}]\label{locking1}
There exists a $\nu_v$-multiflow that locks all the elementary regions $A_i$.
\end{theorem}
\noindent A slightly stronger theorem was also proved, namely that all the elementary regions together with any single composite region can be simultaneously locked. In this section we will generalize these theorems to any cross-free set of regions:
\begin{theorem}[Weak continuum locking] \label{WCLT}
Let $\I\subseteq2^{[n]}$ be cross-free. Then there exists a $\nu_v$-multiflow that locks $\I$.
\end{theorem}
We will give two proofs of Theorem \ref{WCLT}. The first proof is  parallel to the proof of Theorem \ref{locking1} given in \cite{Cui:2018dyq}, and is based on the dual program \eqref{dual2}. For the second proof we will instead present an algorithm for constructing a multiflow via a process we refer to as bulk-cell decomposition, applying Theorem \ref{locking1} within each cell.

In both proofs it will be convenient to assume that $\I$ does not include the empty set and that no pair $I,J\in\I$ obeys
\begin{equation}\label{nocover}
I\cup J = [n]\,.
\end{equation}
These assumptions are without loss of generality: the empty set is locked by any multiflow. And the assumption \eqref{nocover} can be imposed by replacing any $I\in\I$ containing $n$ with $\bar I$, which leaves invariant the locking properties of any given multiflow. With these assumptions, we then have from \eqref{notcrossing} that any distinct pair $I,J\in\I$ is either disjoint or nested, i.e.\ exactly one of the following holds:
\begin{equation}\label{alternatives}
I\cap J=\emptyset\quad\text{or}\quad I\subset J\quad\text{or}\quad J\subset I\,.
\end{equation}

\subsection{Proof by dualization}

In order to prove the Weak Continuum Locking Theorem using duality, we will follow very closely the proofs of the theorems of \cite{Cui:2018dyq}. We therefore briefly review the proof of Theorem \ref{locking1}.

\begin{proof}[Proof of Theorem \ref{locking1}]
With $\I = \{\{1\},\ldots,\{n\}\}$, according to the definitions \eqref{nijdef}, \eqref{Nijdef}, we have $n_{ij}=1$, hence $N_{ij}=2$, for all $i\neq j$. Given a feasible function $\lambda$ for the dual program \eqref{dual2}, we have the function $\phi_i$ on $\M$ defined in \eqref{phiidef}. Using this we define the bulk set $R_i$:
\begin{equation}
R_i:=\{x\in\M|\phi_i(x)<1\}\,.
\end{equation}
By the condition \eqref{eq:pathBound}, $R_i\cap R_j=\emptyset$ for $i\neq j$. By the condition \eqref{positive}, we then have the following lower bound on the dual objective:
\begin{equation}
\int_{\M}\sqrt g\,\lambda \ge \sum_i\int_{R_i}\sqrt g\,\lambda\,.
\end{equation}
Defining the bulk region
\begin{equation}
r_i(p):=\{x\in\M|\phi_i(x)\le p\}
\end{equation}
for $0<p<1$, the level set
\begin{equation}
m_i(p) := \partial r_i(p)\setminus A_i
\end{equation}
is contained in $R_i$ and, being homologous to $A_i$, obeys
\begin{equation}\label{levelsetbound}
\area(m_i(p))\ge S(A_i)\,.
\end{equation}
We then have
\begin{equation}\label{coarea}
\int_{R_i}\sqrt g\,\lambda = \int_{R_i}\sqrt g\,|\nabla\phi_i| = \int_0^1 dp\,\area(m_i(p)) \ge S(A_i)\,,
\end{equation}
where in the first equality we used the fact that $|\nabla\phi_i|=\lambda$; in the second the coarea formula; and in the third \eqref{levelsetbound}. So we have the following lower bound for the dual objective functional:
\begin{equation}\label{dualbound}
\int_{\M}\sqrt g\,\lambda \ge \sum_iS(A_i)\,.
\end{equation}
On the other hand, according to \eqref{total flux bound}, the objective---and therefore the optimal value---of the primal \eqref{primal} is bounded above by $\sum_iS(A_i)$. Hence \eqref{total flux bound} is tight and, as explained below that inequality, it follows that all of the $A_i$s can be simultaneously locked.
\end{proof}

The proof of Theorem \ref{WCLT} follows the same general logic, but involves a bit more bookkeeping to define the bulk region corresponding to each boundary region.

\begin{proof}[Proof of Theorem \ref{WCLT} by dualization]

Again we start with the primal-dual pair \eqref{primal}, \eqref{dual2}, and again we will show that, for any feasible function $\lambda$, the dual objective is bounded below by the sum of the minimal-surface areas of the $A_I$s:
\begin{equation}\label{boundedbelow}
\int_{\M}\sqrt g\,\lambda \ge \sum_{I\in\I} S(A_I)\,.
\end{equation}
This implies that the dual optimal value, and therefore the primal optimal value, are likewise bounded below. Therefore \eqref{total flux bound} is tight, and the $A_I$s can be simultaneously locked.

In order to prove \eqref{boundedbelow}, we will construct for each $A_I$ a scalar function with level sets homologous to it, and show that the dual objective is bounded below by the sum over $I$ of the average level set area, which is in turn bounded below by $\sum_{I\in\I} S(A_I)$. To begin, define for each $i$ and $I\in\I$ the number $N_{iI}$,
\begin{equation}\label{Nialphadef}
N_{iI}:=|\{J\in \I|i\in J \subset I\}|\,,
\end{equation}
which counts the number of sets $J\in\I$ that contain $i$ and are  contained in $I$.
From this definition, it follows that 
\begin{align}
    N_{iI}&\leq n_{ij}-1\,,\quad i\in I,j\in \bar I \label{eq:24}
\\
    N_{iI}&=0\,,\;\;\quad\qquad i\in\bar I\,,
\end{align}
where $n_{ij}$ is defined in \eqref{nijdef}, and the first inequality relies on the fact that $I$ is not counted in $N_{iI}$.

We now employ the functions $\phi_i$ defined in \eqref{phiidef}, and note that the constraint \eqref{eq:pathBound} implies 
\begin{equation}
    \phi_{i}(x)+\phi_{j}(x) \geq N_{ij}\,.
    \label{eq:26}
\end{equation}
Based on these functions, we define a similar function associated with each composite region:
\begin{equation}
    \phi_I(x) := \min_{i\in I} (\phi_{i}(x)-N_{iI})\qquad(I\in\I)\,.
    \label{eq:27}
\end{equation}
$\phi_I$ is continuous, and its gradient at any given point is  equal to the gradient of $\phi_{i_{\rm min}}$, where $i_{\rm min}$ is the minimizer in \eqref{eq:27}, which in turn has magnitude $\lambda$:
\begin{equation}
|\nabla\phi_I| = |\nabla\phi_{i_{\rm min}}| = \lambda\,.
\end{equation}

We will now show that $\phi_I\le0$ on $A_I$ and $\phi_I\ge1$ on $A_{\bar I}$. For $x\in A_{i\in I}$, $\phi_{i}(x)=0$ so $\phi_I(x)\leq 0$. On the other hand, for $x\in A_{j\in\bar I}$, let $i$ be the minimizer of \eqref{eq:27}. Since $\phi_{j}(x)=0$, $\phi_i(x)\ge N_{ij}$.  Thus by \eqref{eq:24},
\begin{equation}
    \phi_I(x) = \phi_{i}(x)-N_{iI} \geq N_{ij}-N_{iI} \geq 1+n_{ji} \geq 1\,.
\end{equation}
This implies that the level set of $\phi_I$ for any value of the function between 0 and 1 is homologous to $A_I$. Therefore, by the same reasoning as in \eqref{coarea}, we have
\begin{equation}    \label{eq:30}
\int_{R_I}\sqrt g\,\lambda\ge S(A_I)\,,
\end{equation}
where
\begin{equation}
R_I:=\{x\in\M\; | 0< \phi_I(x)< 1 \}\,.
\end{equation}
It remains to show that the sets $R_I$ do not intersect. Then we will have
\begin{equation}
\int_{\M}\sqrt g\,\lambda\ge\sum_{I\in\I}\int_{R_I}\sqrt g\,\lambda\,,
\end{equation}
and \eqref{boundedbelow} follows immediately from \eqref{eq:30}.

To show that $R_I \cap R_J=\emptyset$ for $I\neq J$, we must consider two cases, depending on whether $I$ and $J$ are disjoint or nested.
\begin{enumerate}
    \item First, suppose $I\cap J=\emptyset$. At a given point $x\in\M$, let $i\in I$ be a minimizer of \eqref{eq:27} and $j\in J$ a corresponding minimizer for $\phi_J$. Then, by \eqref{eq:24},
   \begin{equation}
        N_{iI}+N_{jJ}\leq N_{ij}-2\,.
        \label{eq:33}
    \end{equation}
Along with \eqref{eq:26} and \eqref{eq:27}, this implies
    \begin{equation}
         \phi_I(x)+\phi_J(x)=\phi_{i}(x)+\phi_{j}(x)-(N_{iI}+N_{jJ})\geq 2\,.
             \end{equation}
Hence $x$ cannot be both in $R_I$ and $R_J$.
    \item Suppose now that $I\subset J$. Let $x\in {R}_I$, and again let $i\in I$ be a minimizer in \eqref{eq:27}. 
    By the definition \eqref{Nialphadef},
    \begin{equation}
        N_{iJ}\ge N_{iI}+1 \,.
    \end{equation}
Along with the fact that $\phi_I(x)<1$, this implies 
    \begin{equation}
    \phi_J(x)\le
    \phi_{i}(x) - N_{iJ} = 
    \phi_I(x) + N_{iI} -N_{iJ}  <
    0\,,
    \end{equation}
    so $x\not\in R_J$.
\end{enumerate}
\end{proof}

\subsection{Proof by decomposition}

In addition to the proof above, which relies on computing the dual program and constraining its optimal solution, we provide a second proof. Rather than relying on the formal theory of convex programming, this is (partially) a proof by construction. We present a method, based on Theorem \ref{locking1}, for constructing a multiflow with maximal flux through all the minimal surfaces $m(A_I)$. To avoid confusion regarding orientation, we proceed by talking mainly in the language of bit threads, knowing that once we have finished constructing a thread configuration it can be converted into a proper multiflow, through the procedure discussed in section \ref{subsection:Multiflows_Basics}. 

\begin{figure}
    \centering
    \includegraphics[width=5in]{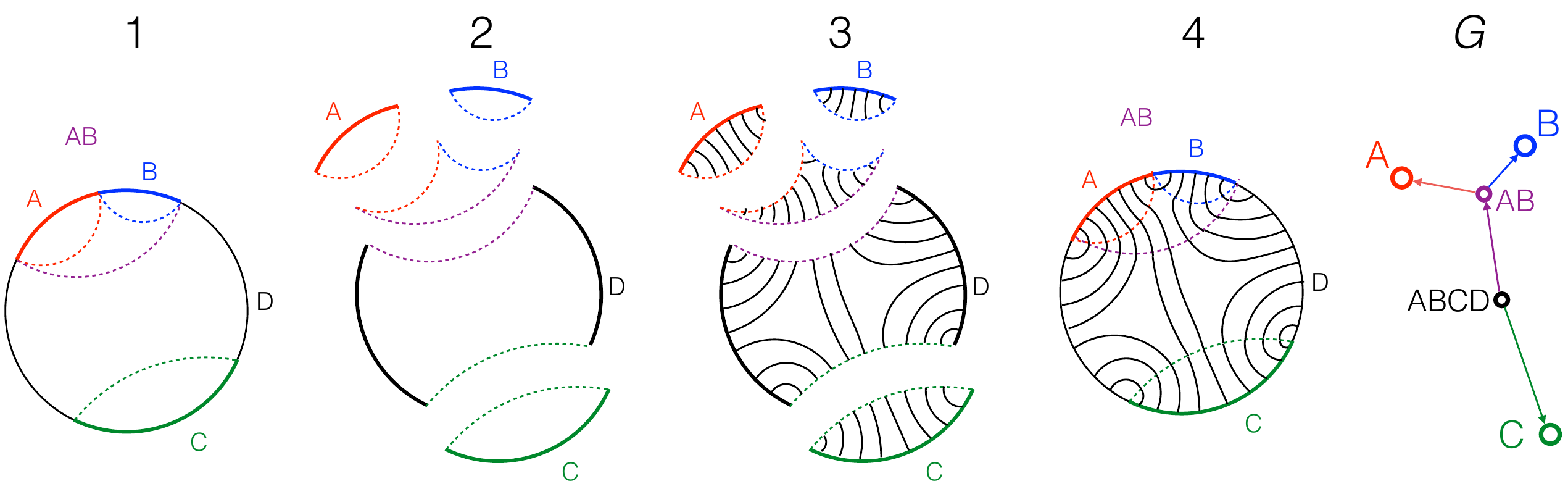}
    \caption{Illustration of the construction of the thread configuration used in the proof by decomposition for the set of boundary regions $\{A,B,AB,C\}$ ($\I=\{\{1\},\{2\},\{1,2\},\{3\}\}$). Step 1: identify minimal surfaces $m(A_{I\in\I})$. Step 2: divide the bulk into cells. Step 3: find a maximal thread configuration for each bulk cell. Step 4: recombine the bulk cells to obtain a maximal thread configuration for the original problem. The graph $G$ used in the text, whose vertices correspond to cells and edges to minimal surfaces, is shown at right for this example.}
    \label{fig:Decomposition}
\end{figure}

We refer to the procedure used to make this construction as ``bulk cell decomposition.'' The basic idea behind this decomposition is to divide the bulk into separate volumes which are bounded by some combinations of the minimal surfaces $m(A_I)$ and the $A_I$ themselves. We can then construct a thread configuration in each cell that maximizes the flux through all of its bounding surfaces. Finally, we recombine these thread configurations along the minimal surfaces that separate them in order to obtain a global thread configuration. See figure \ref{fig:Decomposition} for an illustration of the method.

\begin{proof}[Proof of Theorem \ref{WCLT} by cell decomposition] Before defining the bulk cell decomposition we need to lay some technical groundwork. First, we adjoin to $\I$ the full set $[n]$, corresponding to the entire boundary $\partial\M$; henceforth in the proof $[n]\in\I$. We then construct a directed graph $G$ with vertices labelled by $I\in\I$, and in which there is an edge from $I$ to $J$ (written $I\to J$) if $I\supset J$ and there is no $K\in\I$ such that $I\supset K\supset J$. A vertex may have any number of edges coming out of it, but has exactly one edge coming into it, except $[n]$, which has no edges coming into it. It follows that $G$ is a connected tree. In particular, there is exactly one path connecting any pair of vertices.

The minimal surfaces $m(A_I)$ for different $I$ may partly coincide. We will give an example and deal with this possibility below. For now, for ease of presentation, we assume that the $m(A_I)$ do not coincide.

With these preliminaries out of the way, we now define the cell decomposition. For each $I\in\I$ we define the bulk region $\M_I$, which we call a ``cell'', as follows:\footnote{More precisely, $\M_I$ is the closure of the right-hand side of \eqref{calphadef}. Similarly, the first surface in \eqref{boundarylist} should be enclosed in a ``closure''.}
\begin{equation}\label{calphadef}
\M_I := r(A_I)\setminus\left(\bigcup_{J\leftarrow I}r(A_J)\right).
\end{equation}
The collection $\{\M_I\}_{I\in \I}$ is a decomposition of $\M$, in the sense that $\M_I\no \M_J$ for $I\neq J$ and $\cup_{I\in\I} \M_I = \M$. Two cells $\M_I$, $\M_J$ adjoin each other if and only if $I$ and $J$ share an edge in $G$.

The boundary of $\M_I$ is the union of the following surfaces:
\begin{equation}\label{boundarylist}
A_I\setminus\left(\bigcup_{J\leftarrow I}A_J\right),\qquad
m(A_I)\,,\qquad
m(A_J)\quad (J\leftarrow I)\,.
\end{equation}
Regarding $\M_I$ as a ``bulk'' in its own right, and considering each surface in \eqref{boundarylist} as an ``elementary boundary region'' for that ``bulk'', we now invoke Theorem \ref{locking1} to assert the existence of a $\nu_v$-multiflow within $\M_I$ locking all of these surfaces. The surface $m(A_I)$ is its own minimal surface, so in the corresponding thread configuration the threads intersect $m(A_I)$ orthogonally and saturate the density bound there, and no thread intersects it twice (see the necessary and sufficient condition for a thread configuration to lock a region given below \eqref{mfmc}); similarly for $m(A_J)$ for $J\leftarrow I$. We remove any closed threads or threads with both endpoints on the first surface in \eqref{boundarylist}.

If we consider the thread configurations constructed as above on two neighboring cells, the threads are orthogonal to their shared minimal surface and saturate the density bound on both sides of it, and therefore can be connected across it to form a valid thread configuration on their union. Since the $\{\M_I\}_{I\in\I}$ form a decomposition of $\M$, we can join the thread configurations in all the cells to construct a thread configuration covering $\M$.

It remains to be shown that this thread configuration locks $\I$, and that it can be converted back into a $\nu_v$-multiflow. A given thread maps to a continuous path on the graph $G$. Within any cell it traverses, it connects distinct ``elementary boundary regions'' (i.e.\ different surfaces in \eqref{boundarylist}). Therefore the path on $G$ never doubles back on itself. Since $G$ is a tree, it follows that the path never crosses the same edge twice; so the thread never crosses the same minimal surface twice. Since the density bound is saturated on every minimal surface, every $A_I$ is locked (again by the condition below \eqref{mfmc}).

To convert the thread configuration back to a multiflow, consider the thread bundle consisting of all the threads connecting $A_i$ and $A_j$ ($i\neq j$). All these threads map to the same path in $G$. Within each cell that they traverse, they are part of the same bundle, i.e.\ they connect the same pair of ``elementary boundary regions'', and therefore (since the thread configuration within the cell was born as a $\nu_v$-multiflow) they are locally parallel. Since they are locally parallel everywhere, they can be written as a flow $\vv_{ij}$. These component flows clearly satisfy the $\nu_v$ norm bound everywhere.

We now address the possibility that the minimal surfaces $m(A_I)$, $m(A_J)$ for distinct $I,J\in\I$ may partly coincide. As an example, let  $n=3$ and $\I = \{A,AB,ABC\}$; if $A$ and $B$ are separated far apart then $r(AB)=r(A)\cup r(B)$ and $m(AB) = m(A)\cup m(B)$. To avoid bookkeeping headaches when we track threads through the bulk, we introduce infinitesimal gaps between coinciding minimal surfaces, arranged to respect the nesting property of homology regions.\footnote{Specifically, suppose $m_{IJ}:=m(A_I)\cap m(A_J)\neq\emptyset$ for some $I\neq J$. By the nesting property of homology regions, for every $K$ on the path $p_{IJ}$ on $G$ connecting $I$ and $J$, $m(A_K)\supset m_{IJ}$, except $K=[n]$. We introduce an infinitesimal gap between $m(A_K)$ and $m(A_L)$ for the distinct pair of vertices $K,L$ on $p_{IJ}$ if (1) $K\leftarrow J\neq[n]$, in which case the gap is included in $r(A_M)$ for all $M\supseteq L$; or (2) $K\leftarrow[n]\rightarrow L$, in which case the gap is only included in $r(A_{[n]})=\M$.} For instance, in the above example we put a gap between $m(A)$ and $m(AB)$ and one between $m(B)$ and $m(AB)$, and these gaps are included in $r(AB)$ and $r(ABC)$ but not $r(A)$ or $r(B)$; the cell $\M_{AB}$ then consists of these two gaps. The construction of the thread configuration then goes through as described above, with the threads going straight across the infinitesimal gaps.
\end{proof}

\section{\boldmath Non-crossing regions: \texorpdfstring{$\nu_c$}{𝜈_c} locking theorem}
\label{section:NOLT}

In the previous section, we proved (by two different methods) that any crossing-free set $\I\subseteq2^{[n]}$ can be locked by a $\nu_v$-multiflow. In this section we will upgrade this theorem by replacing the $\nu_v$ collective norm bound with the stronger $\nu_c$ bound. We remind the reader that the $\nu_c$ bound, in addition to limiting the norm of each component flow to 1 everywhere, forbids different components from having overlapping supports (see \eqref{nucdef} for the definition of the $\nu_c$ norm bound). We thus have:

\begin{theorem}[Non-overlapping locking] \label{NOLT}
Let $\I\subseteq2^{[n]}$ be cross-free. Then there exists a $\nu_{c}$-multiflow that locks $\I$.
\end{theorem}

To prove Theorem \ref{NOLT}, we will start from a $\nu_v$-multiflow that locks $\I$, whose existence is guaranteed by Theorem \ref{WCLT}, and give an algorithm for converting it into a $\nu_c$-multiflow without changing the flux out of each $A_I$. The algorithm iterates through pairs of overlapping component flows, removing the overlap (which we call ``untangling''). It does this by adding (or subtracting) the two component flows, replacing them with a single combined flow. We show that, for non-crossing boundary regions, it is always possible to do this while preserving the flux out of all boundary regions. The individual threads in this combined flow are then re-distributed to the appropriate component flows according to their endpoints.

In subsection \ref{subsection:theUntanglingProcess}, we introduce a simple example to demonstrate the concept, prove that the algorithm is always possible and well-defined given the assumptions of a non-crossing multiflow problem, and present a precise definition of the untangling operator. Subsection \ref{subsection:retanglingConvergence} is concerned with the performance of the untangling operator, and how to use it to fully convert a $\nu_v$-multiflow into a $\nu_c$-multiflow. We first present the primary issue with any untangling algorithm, namely that it is not guaranteed to complete in a finite number of untangling steps. We then show that any such infinite sequence of multiflows must converge to a $\nu_c$-multiflow.

\subsection{The untangling operator}\label{subsection:theUntanglingProcess}

The proof of Theorem \ref{NOLT} relies on a method of ``untangling'' overlapping component flows by adding (or subtracting) them and re-separating the sum into new component flows based on the flow line endpoints. The aim of this subsection is to define, for any distinct index pairs $(k,l)\neq(m,n)$, an operator $U_{kl,mn}$ that takes a $\nu_v$-multiflow $V$ and return a $\nu_v$-multiflow $V^2=U_{kl,mn}[V]$ in which the components $\vv^2_{kl}$, $\vv^2_{mn}$ have non-overlapping support. Furthermore, the operator preserves the total flux of the multiflow out of each $A_I$, so if $V$ locks $\I$ then so does $V^2$.

\begin{figure}
	\centering
    \includegraphics[width=2.75in]{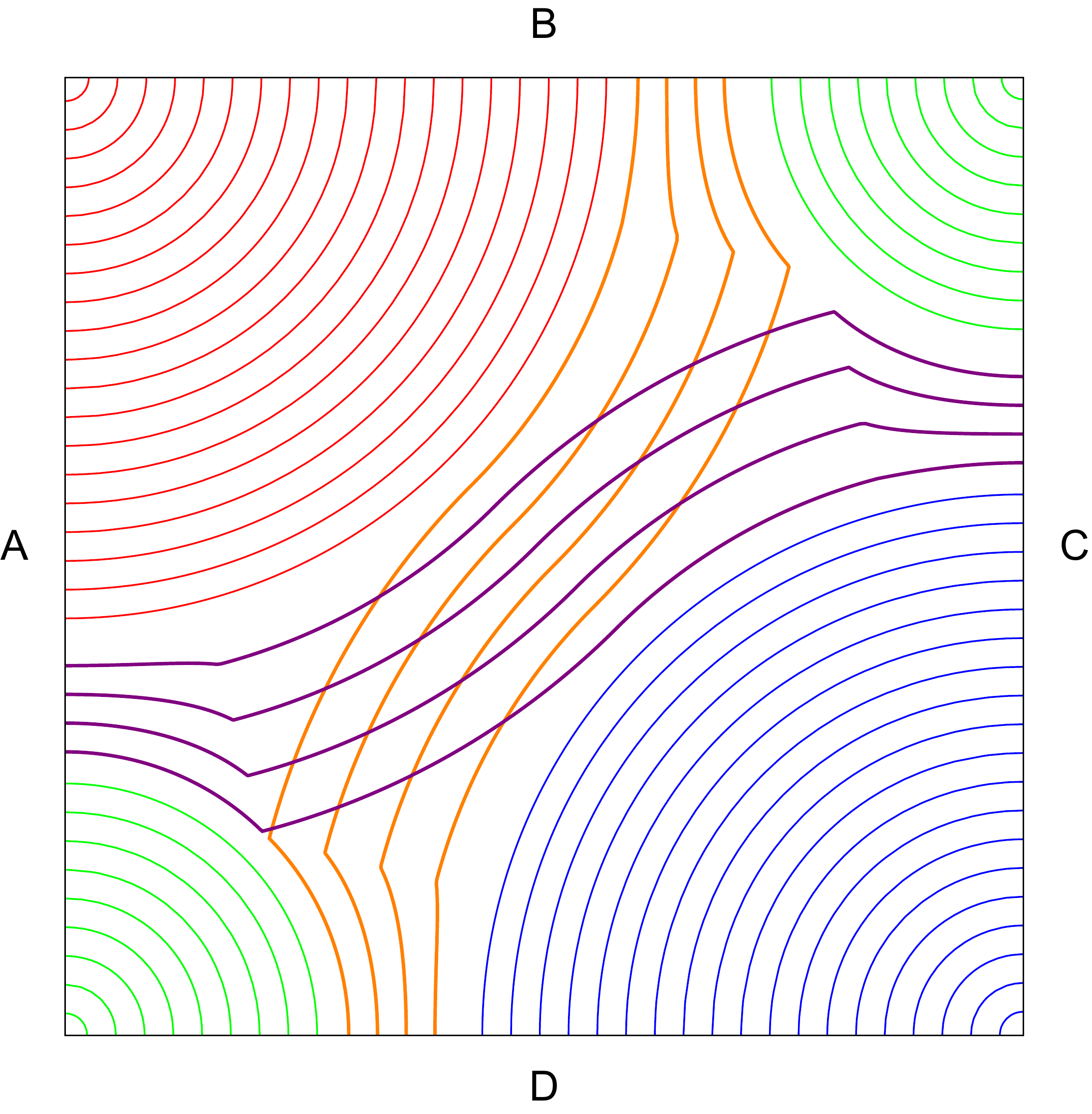}
    \includegraphics[width=2.75in]{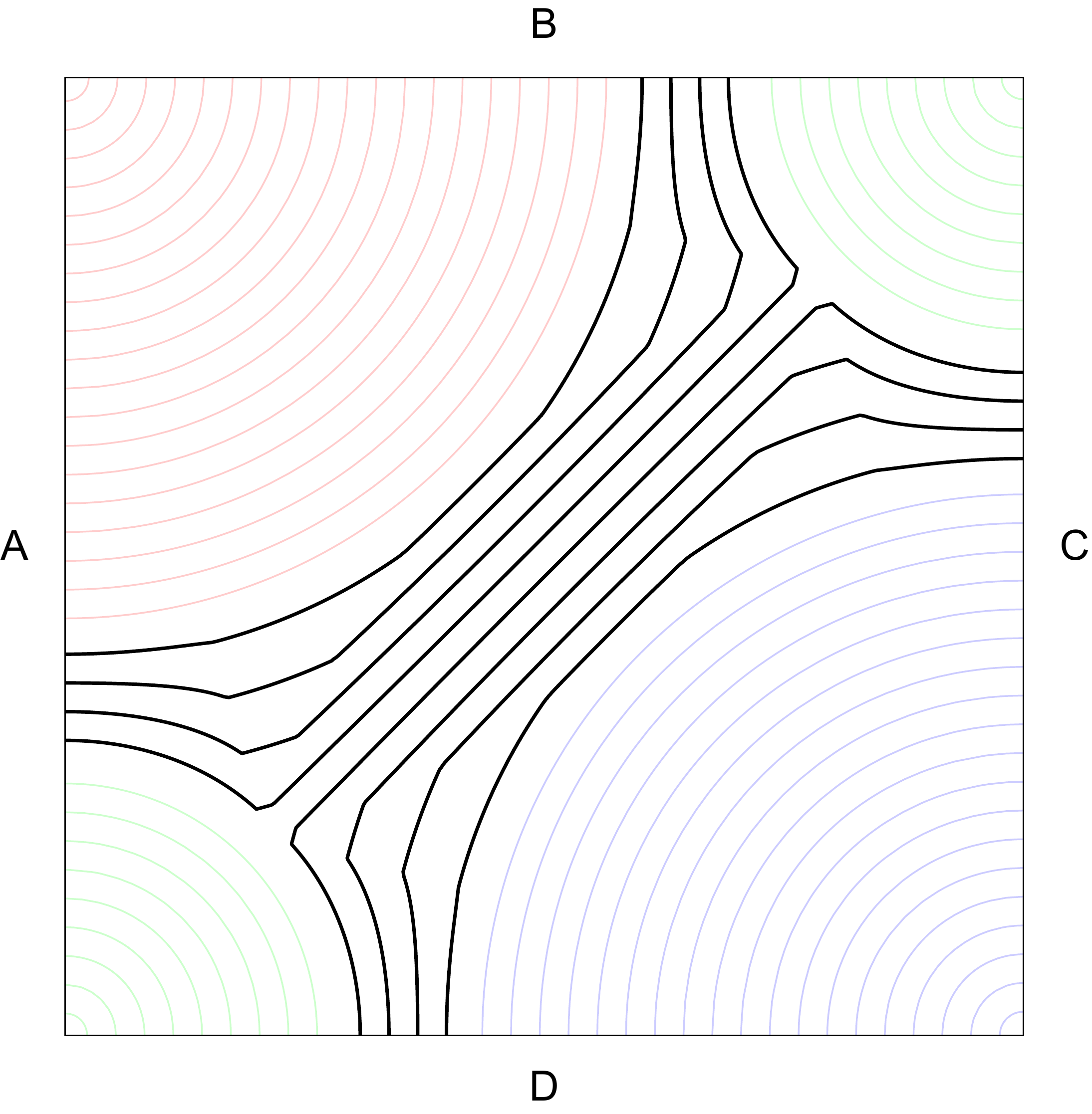}
    \includegraphics[width=2.75in]{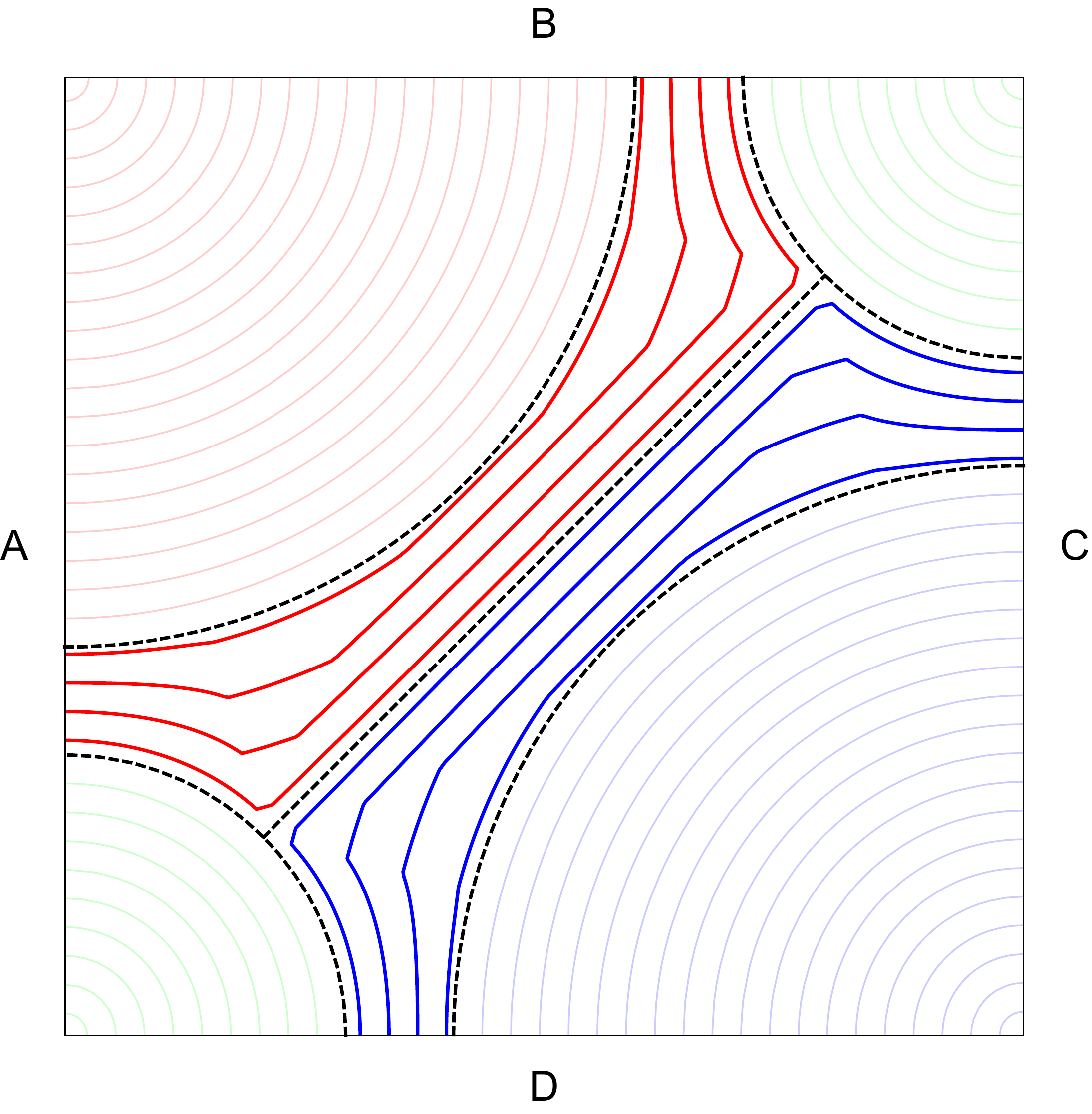}
    \includegraphics[width=2.75in]{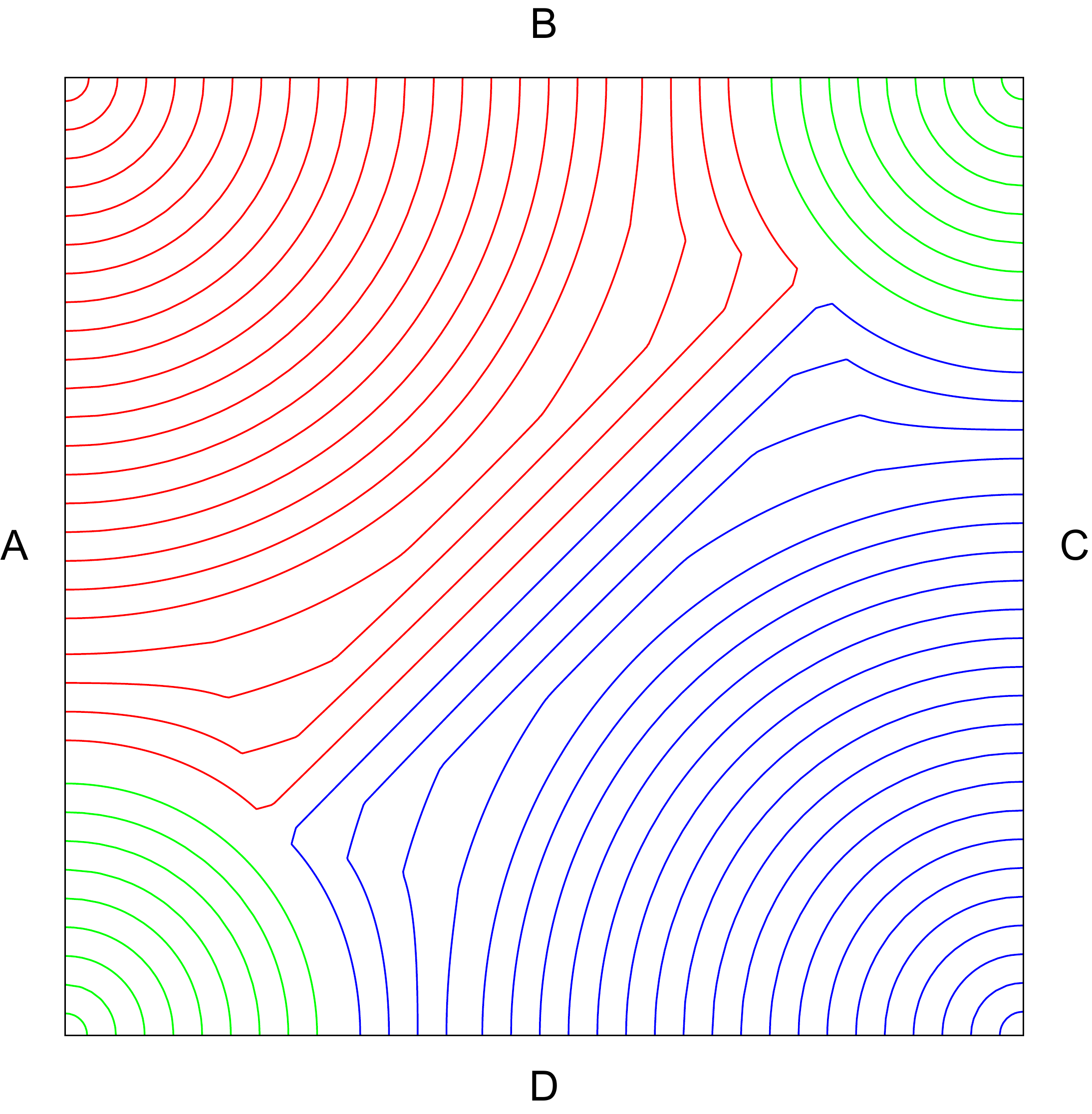}
    \caption{Step-by-step guide to the untangling process, shown on the unit square with sides $A,B,C,D$ and $\{A_I\}_{I\in\I}=\{A,B,C,D,AD\}$. Top left: the original $\nu_v$-multiflow. Two component flows, $\vec{v}_{AC}$ (purple) and $\vec{v}_{BD}$ (orange) overlap. Top right: The two crossing flows are combined together. In order to guarantee that their fluxes out of AD combine constructively, $\vec{v}_{BD}$ is subtracted from $\vec{v}_{AC}$ to produce the black flow, $\vec{v}'$. Bottom left: The combined flow, $\vec{v}'$, is separated into components based on thread endpoints. Stream-tubes (black dashed) originating at $\partial(A_{i}\cap \mathrm{supp}(\vec{v}'))$ for each primitive boundary region $A_i$ visualize this partition into components. In this case, $\vec{v}'$ separates into $\vec{v}_{AB}'$ and $\vec{v}_{CD}'$. Bottom right: $\vec{v}_{AB}'$ and $\vec{v}_{CD}'$ are combined with their unprimed counterparts, producing a new multiflow without overlap.}
    \label{fig:untanglingExample}
\end{figure}

\paragraph{Toy example.}
We begin with a simple example based on a toy problem (see section \ref{section:crossing} for a more in-depth description of this problem). Let $\mathcal{M}$ be the unit square with flat metric, the elementary regions be its sides, and the composite regions we are trying to lock be $\{A_I\}_{I\in\I}=\{A,B,C,D,AD\}$. It is easy to see that this problem does not involve any crossing regions.
We begin with a $\nu_v$-compliant solution to the problem, as guaranteed by the Weak Continuum Locking Theorem and seen in figure \ref{fig:untanglingExample}a. To better demonstrate the untangling process, we have deliberately picked a starting multiflow with overlapping component flows. Here, $\vec{v}_{AC}$ overlaps with $\vec{v}_{BD}$. To untangle them, our first step is to replace these two component flows with a new vector field, $\vec{v}'=\vec{v}_{AC}-\vec{v}_{BD}$. This combined field maintains the flux out of all composite boundary regions (fig \ref{fig:untanglingExample}b). There remains the task of separating $\vec{v}'$ back into component flows, allowing the process to repeat as necessary. For this purpose, the vector field is converted into a set of threads, which are sorted based on their endpoints and assigned back to the appropriate component flow. At the conclusion of the process, we are left with a new multiflow that still locks all the given composite regions, but without overlapping component flows.

\paragraph{Conservation of flux.}
We now return to the general case. Let $\vec{v}_{kl}$ and $\vec{v}_{mn}$ be distinct component flows of a $\nu_v$-multiflow $V$. (While distinct,  $\vec{v}_{kl}$ and $\vec{v}_{mn}$ may share a common index; for example $k$ and $m$ could be equal. In this subsection, $k,l,m,n$ should be regarded as fixed, while $i,j$ are variables.)

Recall the definition of $\vv_I$ from \eqref{v from vij}:\begin{equation}\label{xidef2}
\vv_I := \sum_{i<j}\xi^I_{ij}\vv_{ij}\,,\qquad
\xi^I_{ij} = \begin{cases}
+1\,,\quad& i\in I,j\in\bar I \\
-1\,,\quad &i\in\bar I,j\in I \\
0\,,\quad &\text{otherwise}
\end{cases}
\end{equation}
(similarly for $\vv^2_I$). Using the assumption that $I,J\in\I$ do not cross, we have the following lemma:
\begin{lemma}
For any $I,J\in\I$, $\xi^I_{kl}\xi^I_{mn}\xi^J_{kl}\xi^J_{mn}=0$ or $+1$.
\end{lemma}
\begin{proof}
Proof by contradiction. Assume $\xi^I_{kl}\xi^I_{mn}\xi^J_{kl}\xi^J_{mn}=-1$. Say for example that $\xi^I_{kl}=\xi^I_{mn}=\xi^J_{kl}=-\xi^J_{mn}=1$. Then $k\in I\cap J$, $l\in \bar I\cap\bar J$, $m\in I\cap\bar J$, $n\in\bar I\cap J$. Since none of these sets are empty, $I$ and $J$ cross (see \eqref{crossingdef}). Similarly for the other possibilities.
\end{proof}

According to the lemma, there are two possibilities:
\begin{enumerate}
\item For all $I\in\I$, $\xi^I_{kl}\xi^I_{mn}=0$ or $+1$; in this case we define
\begin{equation}
\vv' := \vv_{kl}+\vv_{mn}\,.
\end{equation}
\item For all $I\in\I$, $\xi^I_{kl}\xi^I_{mn}=0$ or $-1$; in this case we define
\begin{equation}
\vv' := \vv_{kl}-\vv_{mn}\,.
\end{equation}
\end{enumerate}
(If, for all $I\in\I$, $\xi^I_{kl}\xi^I_{mn}=0$, we can define $\vv'$ in either way.) This definition guarantees that no flux out of any $A_I$ is lost by combining $\vv_{kl}$ and $\vv_{mn}$ in this way. More precisely, if in the multiflow $V$ we were to replace $\vv_{kl}$ by $\vv'$ and set $\vv_{mn}$ to zero, then $\int_{A_I}\vv_I$ would remain unchanged for all $I\in\I$. Of course, the result would not obey the constraint  \eqref{eq:flow_terminal} (although it would still obey \eqref{eq:divergenceless1} and \eqref{nubound}) and therefore would not constitute a $\nu_v$-multiflow. We will remedy this defect in the next paragraph by separating $\vv'$ into appropriate component flows.\footnote{It may seem tempting to try to untangle more than two flows at once, accelerating us towards our eventual $\nu_c$-compliant solution. However, there is in general no way to combine three or more flows into a composite flow while conserving the flux out of all $A_I$. In fact, the very simple case of three primitive boundary regions (each its own composite region) and three component flows provides a counterexample, which will be left as an exercise to the reader.}

\paragraph{Formalizing the untangling step.}
We now introduce some new notation for the untangling step, giving it a more precise definition. We will define the operator $U_{kl,mn}$, which acts on a $\nu_v$-multiflow $V$ and gives a $\nu_v$-multiflow $V^2$ in which the component flows  $\vv^2_{kl}$ and $\vv^2_{mn}$ are non-overlapping; furthermore, with the sign chosen as explained in the previous paragraph, the flux through $A_I$ is preserved, i.e.
\begin{equation}\label{fluxpreserve}
\int_{A_I}\vv_I^2 = \int_{A_I}\vv_I\,,
\end{equation}
for all $I\in\I$. We assume without loss of generality that $k\le m$. Step by step, $V^2$ is defined by the following procedure:
\begin{enumerate}
\item \textbf{Add} (subtract) $\vec{v}_{kl}$ and $\vec{v}_{mn}$ into a new composite flow, $\vec{v}^{\prime}=\vec{v}_{kl}+\vec{v}_{mn}$ ($\vec{v}^{\prime}=\vec{v}_{kl}-\vec{v}_{mn}$), as explained in the previous paragraph.
\item \textbf{Separate} the composite flow $\vec{v}^{\prime}$ into $5$ distinct sub-flows:
	\begin{enumerate}
	\item $\vec{v}^{\prime}_{kl}$
	\item $\vec{v}^{\prime}_{mn}$ ($\vec{v}^{\prime}_{nm}$)
	\item $\vec{v}^{\prime}_{kn}$ ($\vec{v}^{\prime}_{km}$)
	\item $\vec{v}^{\prime}_{ml}$ ($\vec{v}^{\prime}_{nl}$)
	\item None of the above
	\end{enumerate}
The first two of these flows correspond to the two original flows being combined. The third and fourth are the two crossover flows, and the fifth flow consists only of internal cycles. If we switch back to the discetized bit thread image, we can perform the separation step by categorizing each thread by its endpoints. In the continuum, streamtubes can be used to partition the support of $\vec{v}^{\prime}$ into five disjoint regions, each occupied by one of these sub-flows and pictured in Figure \ref{fig:untanglingExample}.
\item \textbf{Switch signs.} Use the identity $\vec{v}^{\prime}_{ji}=-\vec{v}^{\prime}_{ij}$ as necessary to return all subflows to the convention $i<j$. Delete the ``None of the above'' flow.
\item \textbf{Recombine} the third and fourth primed sub-flows (c and d) with their corresponding unprimed component flows. Return a new multiflow $U_{kl,mn}[V]=V^2$ with the following components:
\begin{equation}
\vec{v}^{2}_{ij}=
\begin{cases}
	\vec{v}^{\prime}_{ij}              &\mbox{if}\:(i,j)\in\{(k,l),(m,n)\}\\
	\vec{v}_{ij}+\vec{v}^{\prime}_{ij} &\mbox{if}\:(i,j)\in\{(k,n),(m,l),(l,m)\}\quad\left(\:\{(k,m),(n,l),(l,n)\}\:\right)\\
	\vec{v}_{ij}                       &\mbox{otherwise}
\end{cases}
\end{equation}
In other words, the two flows that were untangled replace their overlapping predecessors, and any other incidental crossover flows created in the process are recombined into the multiflow by being added to the component flow with the same endpoints.
\end{enumerate}

\subsection{Re-tangling and proof of convergence}\label{subsection:retanglingConvergence}

We have now demonstrated that we can untangle any two component flows. One might naively suppose that we need only perform the algorithm once on every pair of component flows, and then we are done. However, this is not the case. When we combine two component flows $\vec{v}_{kl}$, $\vec{v}_{mn}$, the composite flow $\vec{v}^{\prime}$ will not necessarily re-separate solely into $\vec{v}'_{kl}$ and $\vec{v}'_{mn}$, but may also produce non-zero crossover flows $\vec{v}'_{kn}$ and/or $\vec{v}'_{ml}$. If the component flow $\vec{v}_{kn}$ or $\vec{v}_{ml}$ was previously separated from some other component flow, this new contribution may re-tangle it.

\begin{figure}
\centering
\includegraphics[width=6.0in]{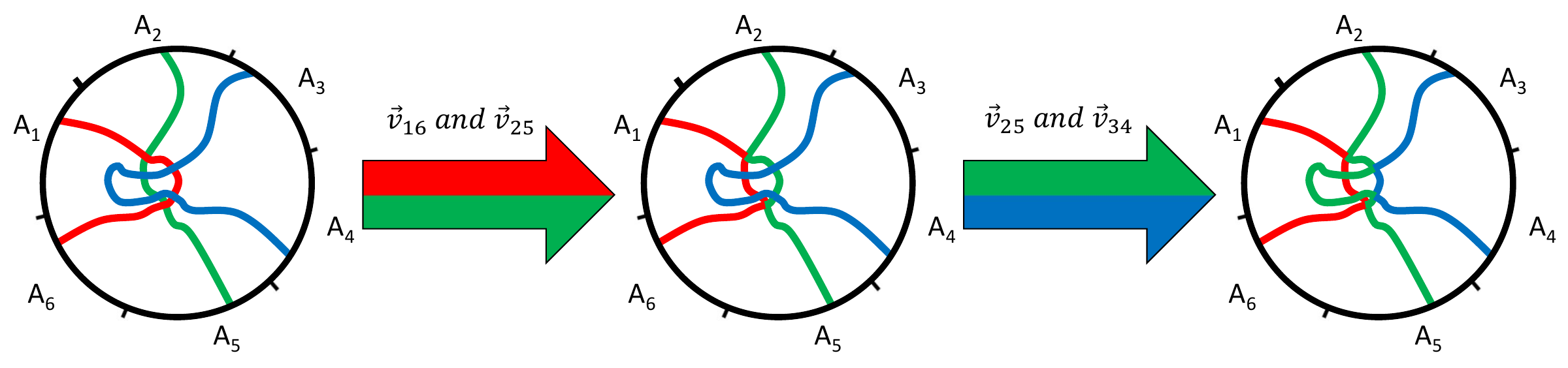}
\caption{A demonstration of the ``re-tangling'' effect. We begin with a multiflow configuration with overlapping pairs of flows. We first untangle the component flows $\vec{v}_{16}$ and $\vec{v}_{25}$. Next we untangle the component flows $\vec{v}_{25}$ and $\vec{v}_{34}$. We note that by doing this we reintroduce an overlap between $\vec{v}_{16}$ and $\vec{v}_{15}$ which we had just untangled.
    }
\label{fig:retangling}
\end{figure}

In fact, this re-tangling effect can take place even when there are no crossover threads. Consider the following example with six boundary regions and three component flows, $\vec{v}_{16}$, $\vec{v}_{25}$, and $\vec{v}_{34}$, visualized in figure \ref{fig:retangling}. We begin by acting with $U_{16,25}$, untangling $\vec{v}_{16}$ and $\vec{v}_{25}$. When we re-separate, there is no crossover, and $\vec{v}_{16}^{'}$ and $\vec{v}_{25}^{'}$ emerge as non-overlapping flows. However, when we next attempt to untangle $\vec{v}_{25}^{'}$ and $\vec{v}_{34}$ with $U_{25,34}$, even without crossover, $\vec{v}_{34}$ transfers its piece that crosses $\vec{v}_{16}^{'}$ back to $\vec{v}_{25}^{'}$, undoing our work.

It is therefore necessary to loop the untangling algorithm through the flows repeatedly until the $\nu_c$ norm bound is satisfied. This could potentially take an infinite number of iterations, so its convergence to a non-overlapping $\nu_c$-multiflow configuration is not a priori obvious. The remainder of the subsection is dedicated to demonstrating this convergence.

First we define a functional $\theta$ on the space of $\nu_v$-multiflows that measures the multiflow's non-compliance with the $\nu_c$ density bound. Specifically,
\begin{equation}
\theta[V]:=\int_{\mathcal{M}}\frac{1}{2}\zeta(\zeta-1),
\end{equation}
where $\zeta=\zeta(V(x))$ is the number of non-zero component flows in $V$ at the point $x\in \mathcal{M}$. The integrand is thus the number of pairs of non-zero component flows at $x$, and $\theta[V]=0$ if and only if $V$ obeys the $\nu_c$ bound everywhere (except on a set of measure zero). To fully define the sequence of multiflows, we now give a more precise definition of what one step in the untangling procedure is.

\paragraph{Iteration operator.}Define one iteration of the untangling algorithm as the untangling of the two flows in the current configuration with the maximum volume of overlap. (While two successive iterations will never untangle the same pair of flows, the algorithm will not necessarily cycle through every pair before repeating a pair.) Denote an iteration by the following operator $U$:\begin{equation}\begin{split}
U[V]:=U_{kl,mn}[V]\,,\quad&\text{where $\vec{v}_{kl},\vec{v}_{mn}$ is the pair of distinct components of $V$} \\ &\text{with maximum volume of support overlap.}
\end{split}\end{equation}

By Theorem \ref{WCLT}, there exists a $\nu_v$-multiflow $V_0$ that locks $\I$. We define the sequence of multiflows and respective $\theta$ values:
\begin{equation}
V_m:= U^{m}[V_0]\,, \quad\theta_m:=\theta[V_m]\,,\quad m=0,1,2,\ldots\,.
\end{equation}
Each $V_m$ is a $\nu_v$-multiflow that locks $\I$. We now prove that the sequence $\theta_m$ is non-increasing and converges to 0.

\begin{lemma}
$\theta_{m+1} \leq \theta_m$.
\label{eq:termleq}
\end{lemma}

\begin{proof}
Note that when two flow fields are locally summed, then only their sum remains in locations in the bulk where they previously overlapped, reducing $\zeta$ in such locations by 1. $\zeta$ does not change in any region where the two flows did not originally overlap. Furthermore, the process of re-separating out the threads back into components also does not change $\zeta$. Thus, $\zeta$ is non-increasing pointwise, so $\theta$ is non-increasing.
\end{proof}

\begin{lemma}
$\lim_{m\to\infty}\theta_m=0$.
\label{eq:convergence}
\end{lemma}

\begin{proof}
We have already shown that the sequence $\theta_m$ is monotonically decreasing and bounded below by zero. This implies that the sequence must converge to some $\theta' \geq 0$. Suppose $\theta'> 0$. Then for any $\epsilon > 0$ there exists some $m_0\in \mathbb{Z}_{+}$ such that
\begin{equation}
\theta_{m_0}-\theta^{\prime}<\epsilon.
\label{eq:epsibound}
\end{equation}
Define $n':=\frac12n(n-1)$, the total number of component flows (where $n$ is the number elementary boundary regions), and $n'':=\frac12n'(n'-1)$, the number of pairs of distinct component flows. Choose
\begin{equation}
\epsilon = \frac{\theta^{\prime}}{n''-1}.
\end{equation}
From \eqref{eq:epsibound} and the choice of $\epsilon$, we have
\begin{equation}
\theta_{m_0} < \theta^{\prime} +\epsilon = \theta^{\prime} \frac{n''
}{n''-1}.
\end{equation}
Since $\theta$ is the integral of the total number of pairs of flows present at each point in the bulk, it can also be thought of as the sum over pairs of component flows of their overlap volume. The operation $U$ always untangles the pair with the maximum overlap volume. This maximum must be greater than or equal to the average volume of overlap across all pairs, so we know that for the maximally overlapping pair of component flows in the configuration $V_{m_0}$,
\begin{equation}
\rho_{\text{overlap}}\geq \frac{\theta_{m_0}}{n''},
\end{equation}
where $\rho_{\text{overlap}}$ is the volume of overlap between the maximally overlapping pair of component flows.

After untangling, $\zeta$ decreases by 1 in the volume of overlap that was affected by the operation, so $\frac{1}{2}\zeta(\zeta-1)$ decreases by at least 1 in that volume (since $\zeta \geq 2$ originally wherever the untangling took place). This means that $\theta$ must decrease by at least $\rho_{\text{overlap}}$ upon untangling, which implies
\begin{equation}
\theta_{m_0+1} \leq \theta_{m_0} -\rho_{\text{overlap}} \leq \theta_{m_0}- \frac{\theta_{m_0}}{n''}=\theta_{m_0}\frac{n''-1}{n''} < \theta^{\prime}\,.
\end{equation}This, in turn, implies that $\theta^{\prime}$ is not a lower bound of the sequence. However, since the sequence is monotonically decreasing, if $\theta^{\prime}$ is not a lower bound then it cannot be the limit point of the sequence, which contradicts our assumption that such a limit point $\theta^{\prime}>0$ exists. Thus the limit point of the sequence must be 0.
\end{proof}

Finally, we now show that the sequence $V_m$ converges pointwise almost everywhere in $\M$ to a $\nu_c$-multiflow. The operator $U$ only changes a multiflow $V$ at points in $\M$ where two or more component flows overlap, i.e.\ where $\zeta\ge2$. Therefore if we let $Z_m$ be the subset of $\M$ where component flows in $V_m$ overlap,
\begin{equation}
Z_m:=\{x\in\M|\zeta(V_m(x))\ge2\}\,,
\end{equation}
then the subsets $Z_m$ are nested:
\begin{equation}
Z_{m+1}\subseteq Z_m\,.
\end{equation}
Hence, except on the limiting set
\begin{equation}
Z_\infty:=\bigcap_{m=0}^\infty Z_m\,,
\end{equation}
$V_m$ converges pointwise to a configuration $V_\infty$ with $\zeta=0$ or 1 everywhere. Furthermore, since the volume of $Z_m$ is bounded above by $\theta_m$, it goes to 0. Hence $Z_\infty$ has vanishing volume, and can be safely neglected. (In particular, since $V_m$ obeys the $\nu_v$ bound for all $m$, all component flow vector fields are finite. So the flux of any component of $V_m$ passing through $Z_\infty$ vanishes.) We conclude that $V_\infty$ is $\nu_c$-multiflow that locks $\I$.

\section{\boldmath Crossing regions: \texorpdfstring{$\nu_v$}{𝜈_v} locking failure}
\label{section:crossing}

Since $\nu_v$-multiflows are capable of locking an arbitrary  arrangement of non-crossing regions, one may ask whether they can also lock crossing regions. Indeed, as mentioned in the introduction (and discussed further in appendix \ref{sec:networks}), on networks there is a theorem guaranteeing that crossing sets of  terminals can be locked under certain conditions. On the other hand, in the continuum there is a potential geometric obstruction to locking for crossing regions under the $\nu_v$ norm bound, which we can most easily state in terms of threads. To lock a region $A_I$, the threads must cross its minimal surface $m(A_I)$ orthogonally and at maximal density. However, the minimal surfaces for two regions that cross often intersect at an angle; at the intersection locus, the threads cannot cross both surfaces perpendicularly.\footnote{This argument was pointed out to us by V. Hubeny.}

\begin{figure}
    \centering
    \includegraphics[width=3in]{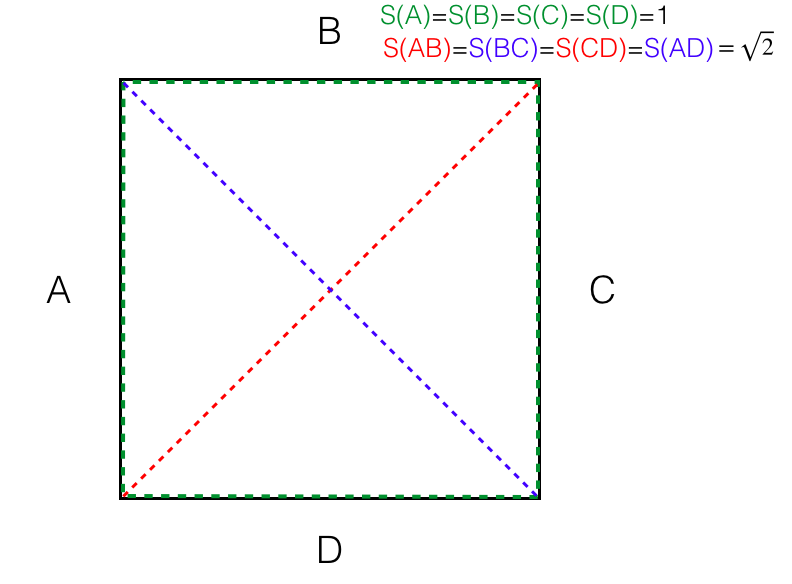}
    \caption{Minimal surfaces and corresponding entropies of the toy problem described in Section \ref{section:crossing}.}
    \label{fig:ToyRT}
\end{figure}

Nonetheless, one might wonder if there could exist some kind of multiflow that is perhaps singular at the intersection locus that locks crossing regions, or a limit of multiflows that comes arbitrarily close to doing so. In order to exclude such possibilities, in this section we will give a relatively simple example in which crossing regions \emph{cannot} be locked by a $\nu_v$-multiflow (or limit of $\nu_v$-multiflows). Specifically, we let $\M$ be the flat unit square and the elementary regions $A,B,C,D$ be its sides in consecutive order (see Figure \ref{fig:ToyRT}). We will attempt to lock the crossing boundary regions $AB$ and $BC$ with a $\nu_v$-multiflow. The minimal surfaces for these regions are the two diagonals respectively, and they have areas $S(AB)=S(BC)=\sqrt2$. The two regions can be simultaneously locked if and only if the maximal total flux equals their total minimal surface area $S(AB)+S(BC)=2\sqrt{2}$. We will study both the primal and dual versions of this problem, \eqref{primal} and \eqref{dual2} respectively. First, we will show that the maximal total flux is less than $2\sqrt2$, establishing that these two regions cannot be locked by a $\nu_v$-multiflow (or by a thread configuration obeying the $\nu_v$ density bound). We will then exhibit numerical solutions to the primal and dual programs.

The primal program \eqref{primal} for this case is
\begin{equation}\label{squareprimal}
\text{Maximize }\int_{AB}\vv_1+\int_{BC}\vv_2
\quad\text{over $\nu_v$-multiflows $V$}\,,
\end{equation}
where
\begin{equation}
\vv_1 = \vv_{AC}+\vv_{BC}+\vv_{AD}+\vv_{BD}\,,\qquad
\vv_2 = -\vv_{AB}-\vv_{AC}+\vv_{BD}+\vv_{CD}\,.
\end{equation}
(In terms of our usual notation, $\vv_1$ and $\vv_2$ are $\vv_I$ for $I=\{1,2\}$ and $I=\{2,3\}$ respectively.) The primal objective can also be written in terms of the fluxes $\int_{A_i}\vv_{ij}$ of the component flows as
\begin{equation}
\int_A\vv_{AB}+2\int_A\vv_{AC}+\int_A\vv_{AD}+\int_B\vv_{BC}+2\int_B\vv_{BD}+\int_C\vv_{CD}\,.
\end{equation}
The dual program \eqref{dual2} is 
\begin{align}
    \label{eq:47}
\text{Minimize $ \int_{\mathcal{M}}\sqrt g\,\lambda$ with respect to $\lambda$, subject to }
    \lambda&\ge0\,,\nonumber\\
   \int_A^ B  ds\,\lambda \:,\int_A^D ds\,\lambda \:,\int_B^C ds\,\lambda \:,\int_C^D ds\,\lambda &\geq 1\,, \\
  \int_A^C ds\,\lambda \:,
  \int_B^D ds\,\lambda &\geq 2\,,\nonumber
\end{align}
where the integrals are over arbitrary curves with endpoints in the given regions.

\begin{figure}
    \centering
    \includegraphics[width=4.5in]{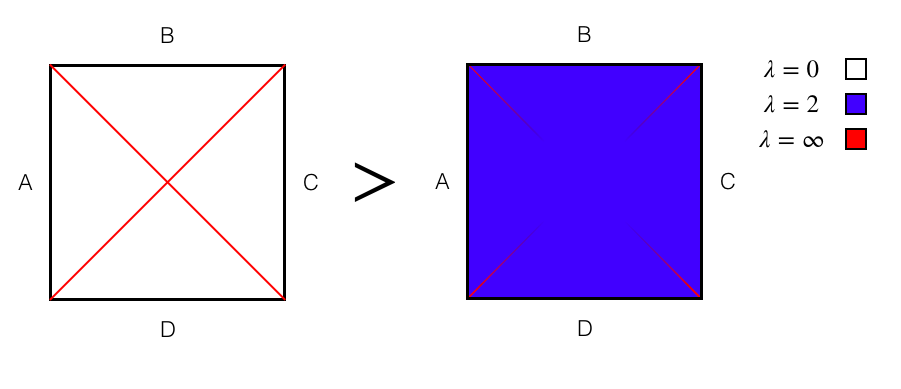}
    \caption{Illustration of feasible configurations for the dual program \eqref{eq:47}. On the left $\lambda$ is zero except for unit-weight delta-functions on the minimal surfaces for $AB$ and $BC$, which are the diagonals of the square. On the right is a feasible $\lambda$ with a smaller value of the objective: $\lambda=2$ everywhere except for delta functions along the diagonals with weight  decreasing linearly with distance from the corner.}
    \label{fig:6}
\end{figure}

By weak duality, any feasible configuration of the dual program provides a rigorous upper bound for the optimal value of the primal program. One such a choice for $\lambda$ is illustrated on the left side of Figure \ref{fig:6}. It consists of a unit-weight delta function along the $AB$ and $BC$ minimal surfaces, which are the diagonals of the square. The dual objective for this configuration takes the value $2\sqrt2$. The question is then whether or not we can do better, i.e.\ achieve a smaller value for the objective while still satisfying the constraints. If so, then the primal objective is bounded above by a value smaller than $2\sqrt2$, and a result we cannot lock both $AB$ and $BC$.

We will now show that this is indeed the case. The configuration, shown on the right side of Figure \ref{fig:6}, is obtained by first setting $\lambda =2$ everywhere on the square to satisfy the constraint on paths between opposite sides ($A$ to $C$ and $B$ to $D$), and then adding weighted delta functions along the diagonals to ensure that the constraint on paths between adjacent sides is also satisfied. The smallest weight for the delta functions to satisfy the constraints is
\begin{equation}
    1-2\sqrt{2}\rho,
\end{equation}
where $\rho$ is the distance from the nearest corner of the square. Note that we terminate the delta functions at a distance $\rho=1/(2\sqrt{2})$ from the corner. Integrating this weight from $\rho=0$ to $1/(2\sqrt{2})$ gives a contribution of $\sqrt{2}/8$ to the objective from each delta function. In all we thus have
\begin{equation}
    \int_{\mathcal{M}}\sqrt g\,\lambda=2+4\frac{\sqrt{2}}{8}=2+\frac{\sqrt{2}}{2}\approx 2.71\,.
\end{equation}
Since this is strictly less than $2\sqrt2\approx 2.83$, we conclude that $AB$ and $BC$ cannot be simultaneously locked by a $\nu_v$-multiflow. By the weak duality \eqref{threadWD} between the dual program and threads, we also conclude that $AB$ and $BC$ cannot be simultaneously locked by any thread configuration obeying the $\nu_v$ density bound. Note that we are not claiming that this configuration is the \emph{solution} to the dual program.

\begin{figure}
    \centering
    \includegraphics[width=5in]{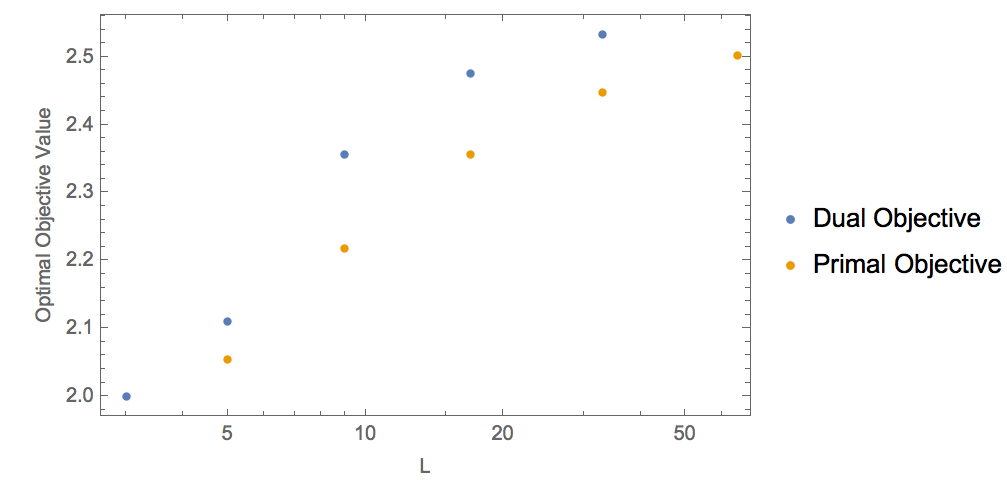}
    \caption{Optimal values of primal and dual programs \eqref{squareprimal}, \eqref{eq:47} as a function of resolution $L$.}
    \label{fig:conv}
\end{figure}

\begin{figure}
    \centering
    \includegraphics[width=4in]{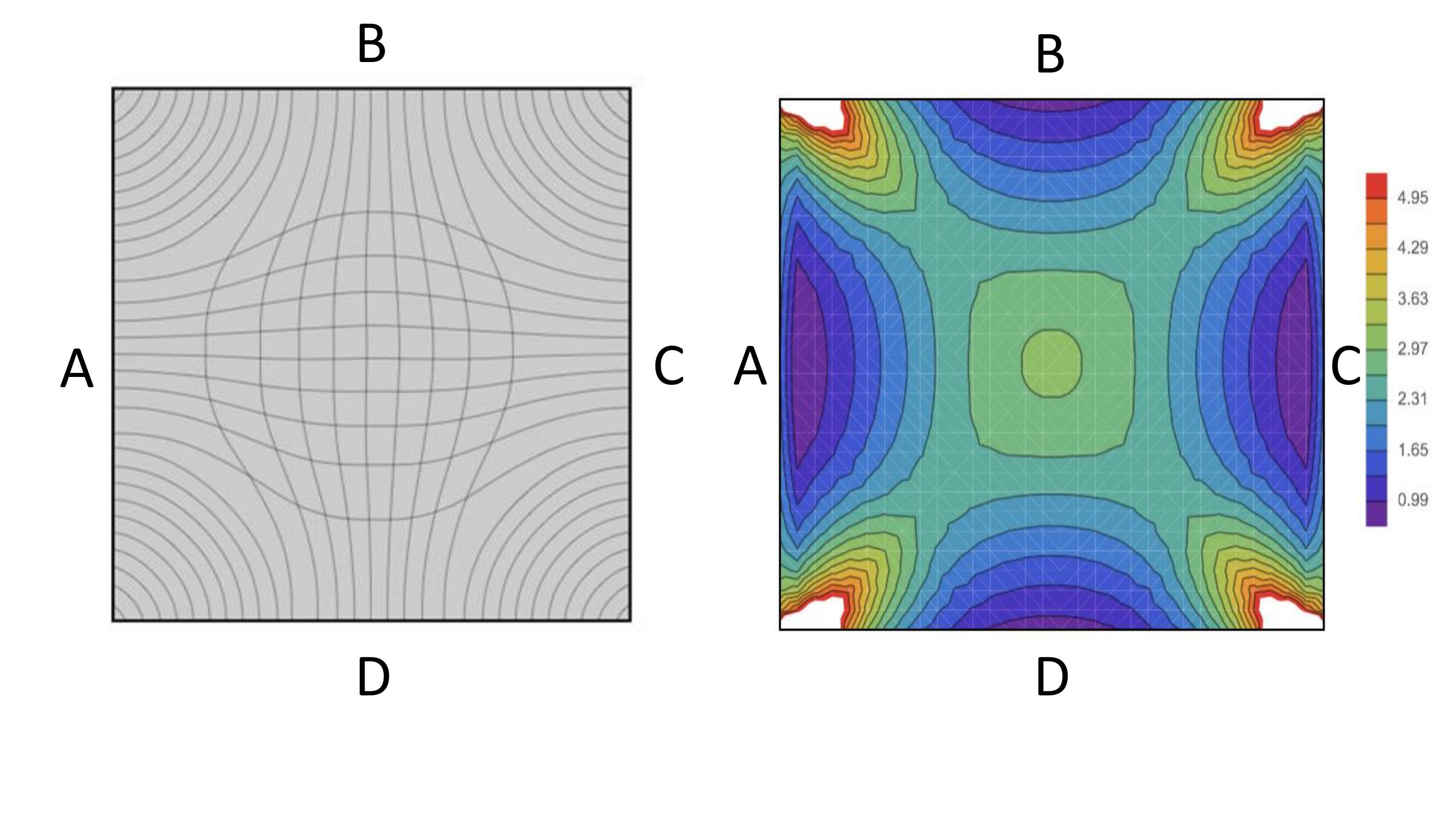}
    \caption{Left: Numerical solution for the primal program \eqref{squareprimal} on an $L=65$ lattice. Shown are a set of flow lines (i.e.\ threads) for each component flow. Right: Contour map of solution $\lambda$ for the dual program \eqref{eq:47} on an $L=33$ lattice.}
    \label{fig:10}
\end{figure}

In view of this result, and in order to better understand the behavior of $\nu_v$-multiflows in the presence of crossing regions, it is interesting to study the solutions to the primal and dual programs. Using \emph{Mathematica}, we performed numerical optimizations for both the primal and dual programs, discretizing the square to an $L\times L$ lattice for various values of $L$. For the primal, we considered resolutions ranging from $L=5$ to $65$, and for the dual ranging from $L=3$ to 33. A plot of the optimal objective values for both the primal and dual computations for different lattice sizes is shown in Figure \ref{fig:conv}. While we did not attempt a serious numerical error analysis, the values shown in Figure \ref{fig:conv} do appear to be converging at large lattices sizes to a value around 2.5 or 2.6, with the primal and dual consistent with each other. The highest-resolution solutions for the primal and dual are shown in Figure \ref{fig:10}. Both appear to be smooth; in particular, there is no evidence of delta functions in the dual solution.

\section{Crossing regions: \texorpdfstring{$\nu_a$}{𝜈\_a} locking theorem and conjectures}
\label{section:altBound}

In the previous section, we showed with an example that crossing regions cannot in general be locked by a $\nu_v$-multiflow. This also certainly holds for $\nu_c$-multiflows, since the $\nu_c$ norm bound is stronger than the $\nu_v$ one. However, in subsection \ref{sec:background}, we discussed a third possible norm bound (or equivalently density bound for threads), that was based on the total flux (or number of threads) passing through any surface. This bound was $\nu_a\le1$, where the norm $\nu_a$ was defined in \eqref{nu1def}; we repeat it here for convenience:
\begin{equation}\label{nu1def2}
\nu_a(V)=\max_{\hat n}\sum_{i<j}|\hat n\cdot\vv_{ij}|\,,
\end{equation}
where the maximum is over unit vectors. The norm $\nu_a$ can also be written in terms of a maximization over a set of scalars $\xi_{ij}\in[-1,1]$ ($i<j$):
\begin{equation}\label{nu1def3}
\nu_a(V)=\max_{\{|\xi_{ij}|\le1\}}\left|\sum_{i<j}\xi_{ij}\vv_{ij}\right|.
\end{equation}
The equivalence of \eqref{nu1def2} and \eqref{nu1def3} can be shown by using the fact that, for any vector $\vv$, $|\vv| = \max_{\hat n}\hat n\cdot\vv$ and switching the order of the maximizations:
\begin{align}
\max_{\{|\xi_{ij}|\le1\}}\left|\sum_{i<j}\xi_{ij}\vv_{ij}\right| &=
\max_{\{|\xi_{ij}|\le1\}}\max_{\hat n}\hat n\cdot\left(\sum_{i<j}\xi_{ij}\vv_{ij}\right)\nonumber \\ &= 
\max_{\{|\xi_{ij}|\le1\}}\max_{\hat n}\left(\sum_{i<j}\xi_{ij}\hat n\cdot\vv_{ij}\right) \nonumber\\ &= 
\max_{\hat n}\sum_{i<j}\max_{|\xi_{ij}|\le1}\xi_{ij}\hat n\cdot\vv_{ij} \nonumber\\ &=\max_{\hat n}\sum_{i<j}|\hat n\cdot\vv_{ij}|\,.
\end{align}

By \eqref{nu1def3}, the condition $\nu_a(V)\le1$ is equivalent to the requirement that any linear combination of component flows,
\begin{equation}
\vv = \sum_{i<j}\xi_{ij}\vv_{ij}\,,
\end{equation}
where the coefficients $\xi_{ij}$ are constants in $[-1,1]$, obeys $|\vv|\le1$ and is therefore itself a flow. In particular, its flux for any boundary region $A_I$ is bounded above by the minimal surface area $S(A_I)$. The $\nu_a$ bound is the weakest possible bound that guarantees this property.

Since the $\nu_a$ bound is less stringent than the $\nu_c$ or $\nu_v$ bounds, it is natural to ask whether a $\nu_a$-multiflow might be able to lock crossing regions (at least in some cases). Let us start with the simple example studied in the previous section, the flat unit square with sides $A,B,C,D$, where we are trying to lock $AB$ and $BC$. Indeed, we find a simple solution within the set of $\nu_a$-multiflows: Set $\vv_{AC}$ to be a constant vector field in the horizontal direction (from $A$ toward $C$) of norm $1/\sqrt2$, and similarly $\vv_{BD}$ a constant vector field in the vertical direction (from $B$ to $D$) of norm $1/\sqrt2$; and set all other component flows to 0. It is easy to check that this multiflow satisfies the $\nu_a$ norm bound, and locks both $AB$ and $BC$.

This encouraging result leads to the following question: Can any pair of crossing regions be locked by a $\nu_a$-multiflow? We turn now to Theorem \ref{continuumcrossing}, which answers this question in the affirmative. Conjectured generalizations will be discussed in subsection \ref{sec:open questions}.

\subsection{Crossing-pair locking theorem}

\begin{figure}
\centering
\includegraphics[width=6in]{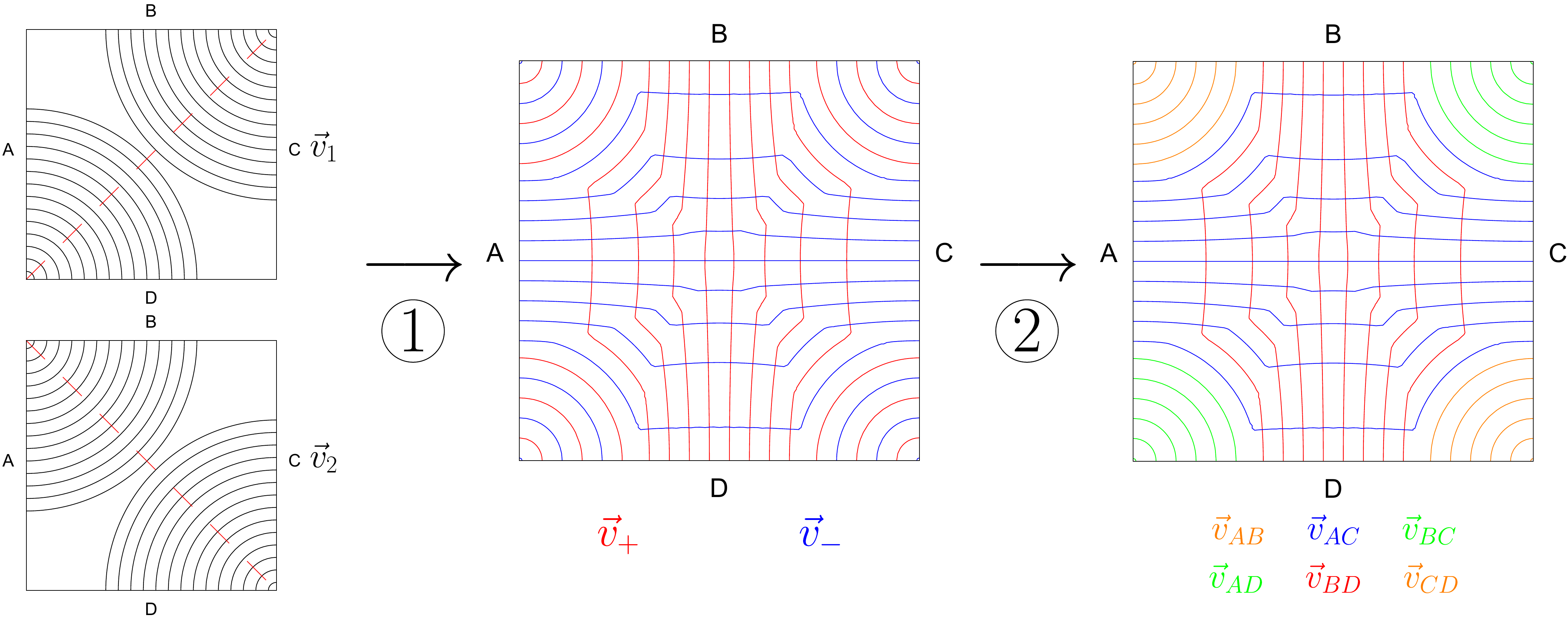}
\caption{Flowchart of the process of converting two flows into a $\nu_a$-multiflow. We start with two separate flows, $\vec{v}_1$ and $\vec{v}_2$, which lock $AB$ and $BC$, respectively (the locked minimal surfaces are shown as red dashed lines). In step (1), the two flows are combined to create $\vec{v}^\pm=\frac{1}{2}(\vec{v}_1\pm\vec{v}_2)$. In step (2), $\vec{v}^\pm$ are separated into components and recombined into the multiflow $V=\{\vec{v}_{ij}\}$.
}
\label{fig:SquareSumDifference}
\end{figure}

In this subsection, we prove the following theorem. The proof is illustrated in Figure \ref{fig:SquareSumDifference} on the example of the square studied in Section \ref{section:crossing}. A straightforward generalization, 
%Matt added phrase
which allows us for example to lock the set $\I=\{A,B,C,D,AB,BC\}$ for $n=4$, %
is given after the proof.

\begin{theorem}[Crossing-pair locking]\label{continuumcrossing}
Let $I,J\subset[n]$ cross. Then there exists a $\nu_a$-multiflow that locks $\{I,J\}$.
\end{theorem}

\begin{proof}
We define the following new set of elementary regions:
\begin{equation}
A':=A_{I\cap\bar J}\,,\quad
B':=A_{I\cap J}\,,\quad
C':=A_{\bar I\cap J}\,,\quad
D':=A_{\bar I\cap\bar J}\,;
\end{equation}
these are non-empty by the fact that $I,J$ cross. Below we will prove the existence of a $\nu_a$-multiflow $V$ for $A',B',C',D'$ that locks $A'B'$ and $B'C'$. The component flows of $V$ can be separated into component flows of the original elemntary regions as described in subsection \ref{sec:multiflow relations} without changing the total flux out of $A_I=A'B'$ and $A_J=B'C'$, and therefore lock $A_I$ and $A_J$. In the rest of the proof we drop the primes on $A,B,C,D$.

The first step of generating a $\nu_a$-multiflow that locks $AB$ and $BC$ is to decompose the multiflow problem into two separate flow problems, each of which can be solved with a single flow. By the max flow-min cut theorem (Theorem \ref{mfmctheorem}), we know that there exist flows $\vv_1$, $\vv_2$ locking $AB$, $BC$ respectively:
\begin{equation}
\int_{AB}\vec{v}_1=S(AB)\,,\qquad
\int_{BC}\vec{v}_2=S(BC)\,.
\end{equation}
From $\vv_{1,2}$, we now define the following pair of flows:
\begin{equation}\label{eq:vpm}
\vec{v}^\pm:=\frac{1}{2}(\vec{v}_1\pm\vec{v}_2)\,.
\end{equation}
Note that, if the two vector fields $\vv^\pm$ were treated as components of a multiflow, then they would obey the $\nu_a$ norm bound, since both their sum and their difference are flows.\footnote{On the other hand, they would not necessarily obey the $\nu_v$ bound, for example if at some point $\vv_{1,2}$ are orthogonal unit vectors. This is thus the step where the $\nu_a$ bound demonstrates its advantage over the $\nu_v$ bound.} Our strategy will be to decompose $\vv^\pm$ into multiflows $V^\pm$, and then combine them (component-wise) into a single multiflow $V$, which we will show obeys the $\nu_a$ bound and locks $AB$ and $BC$.

We turn each flow $\vv^\pm$ into a multiflow $V^\pm$ by separating it into components $\vv_{ij}^\pm$ according to the endpoints of the flow lines (or threads), as described in subsection \ref{sec:multiflow relations}. In that procedure the orientations of the resulting component flows were not specified; here we choose the orientations so that
\begin{equation}
\vv^\pm=\sum_{i<j}\vv_{ij}^\pm\,.
\end{equation}
(Note that with this choice of orientation the flux of $\vv^\pm_{ij}$ is not necessarily positive from $A_i$ and negative from $A_j$.) We also choose not to assign any closed threads in the $\vv^+$ configuration to the $\vv^+_{AC}$ component, or closed threads in $\vv^-$ to the $\vv^-_{BD}$ component; this will be useful in Lemma \ref{vAC0} below. Since they are derived from a single flow, the components $\vv^+_{ij}$ have non-overlapping supports, and similarly for $\vv_{ij}^-$. Thus, each multiflow $V^\pm$ is a $\nu_c$-multiflow.

The multiflows $V^\pm$ are now combined into a single multiflow $V$, taking care to respect the signs of each of the component flows, inherited from the boundary regions they were designated to lock. Specifically, $V$ is defined to have the following component flows:
\begin{equation}\label{eq:newmf2}
\begin{split}
&\vec{v}_{AB}=-\vec{v}^+_{AB}+\vec{v}^-_{AB}\\
&\vec{v}_{AC}=\vec{v}^-_{AC}\\
&\vec{v}_{AD}=\vec{v}^+_{AD}+\vec{v}^-_{AD}\\
&\vec{v}_{BC}=\vec{v}^+_{BC}+\vec{v}^-_{BC}\\
&\vec{v}_{BD}=\vec{v}^+_{BD}\\
&\vec{v}_{CD}=\vec{v}^+_{CD}-\vec{v}^-_{CD}\,.
\end{split}
\end{equation}

Note that two component flows, $\vec{v}^+_{AC}$ and $\vec{v}^-_{BD}$ never show up in this assignment. As it turns out, these two flows are both zero flows by construction.

\begin{lemma}\label{vAC0}
$\vec{v}^+_{AC}=\vec{v}^-_{BD}=0$
\end{lemma}

\begin{proof}
We begin with $\vec{v}^+_{AC}$. Since $\vv_1$ locks $AB$, on the minimal surface $m(AB)$ it equals $\hat n$, the unit normal pointing away from $r(AB)$. Together with the fact that $|\vv_2|\le1$, this implies that $\hat{n}\cdot\vec{v}^+=\frac{1}{2}\hat{n}\cdot\left(\vec{v}_1+\vec{v}_2\right)=\frac{1}{2}\left(1+\hat{n}\cdot\vec{v}_2\right)\geq 0$. In words, at any point on $m(AB)$, $\vec{v}^+$ flow, and therefore all of its components $\vv^+_{ij}$, point out of $r(AB)$ or are zero; no flux can pass into $r(AB)$. An identical argument follows for $m(BC)$, with the roles of $\vec{v}_1$ and $\vec{v}_2$ swapped, proving that no flux can pass into $r(BC)$. However, any flow connecting $A$ and $C$ must either go from $A$ to $C$, passing through $m(BC)$ into $r(BC)$, or from $C$ to $A$, passing through $m(AB)$ into $r(AB)$. Accordingly, the flow $\vec{v}^+_{AC}$ must have no flux, and be $0$ everywhere.

A similar argument can be used to show that the $\vec{v}^-_{BD}$ flow can neither arrive at nor depart from the boundary region $B$, and so it too must vanish.
\end{proof}

We check again that the $\nu_a$ density bound remains obeyed through this separation and recombination. Since $V^\pm$ are both $\nu_c$-multiflows, at any point $x\in\mathcal{M}$ there exists at most one non-zero $\vec{v}^+_{ij}$ and one non-zero $\vec{v}^-_{ij}$. According to \eqref{eq:newmf2}, each $\vec{v}^{\pm}_{ij}$ contributes to at most one $\vec{v}_{ij}$, so an arbitrary sum of distinct $\vec{v}_{ij}$ components also boils down to a sum of at most one $\vec{v}^+_{ij}$ and one $\vec{v}^-_{ij}$ at any $x\in\mathcal{M}$. If a  $\vec{v}^+$ flow and a $\vec{v}^-$ are not both present, the $\nu_a$ bound is trivially satisfied. Suppose $\vec{v}^+_{kl}$ and $\vec{v}^-_{mn}$ are the two non-zero components at $x$. Assuming $\{k,l\}\neq\{m,n\}$,
\begin{equation}
\nu_a(V)=\max_{|\xi_{ij}|\le1}|\xi_{kl}\vec{v}^+_{kl}+\xi_{mn}\vec{v}^-_{mn}|=\max_{+,-}|\vec{v}^+_{kl}\pm\vec{v}^-_{mn}|=\max_{+,-}|\vec{v}^+\pm\vec{v}^-|=\max\{|\vec{v}_1|,|\vec{v}_2|\}\leq 1\,.
\end{equation}
If $\{k,l\}=\{m,n\}$, this simply fixes the sign in the above equation according to the definition of $\vec{v}_{ij}$ given in \eqref{eq:newmf2}, so the bound still holds.

Finally, we verify that the multiflow $V$ does in fact lock  $AB$ and $BC$. Consider the flows $\vec{v}^{\prime}_1$ and $\vec{v}^{\prime}_2$, constructed from the components of $V$ as defined in \eqref{v from vij} to lock $AB$ and $BC$, respectively:
\begin{equation}\label{vprimedef}
\begin{split}
\vec{v}^{\prime}_1&=\vec{v}_{AC}+\vec{v}_{AD}+\vec{v}_{BC}+\vec{v}_{BD}\\
\vec{v}^{\prime}_2&=-\vec{v}_{AB}-\vec{v}_{AC}+\vec{v}_{BD}+\vec{v}_{CD}\,.
\end{split}
\end{equation}
The total flux of $\vv_1'$ is readily calculated:
\begin{align}\label{ABlock}
\int_{AB}\vec{v}^{\prime}_{1}&=\int_{AB}\left(\vec{v}^-_{AC}+\vec{v}^+_{AD}+\vec{v}^-_{AD}+\vec{v}^+_{BC}+\vec{v}^-_{BC}+\vec{v}^+_{BD}\right)
\nonumber\\
&=\int_{AB}\left(\vec{v}^++\vec{v}^-\right) \nonumber\\
&=\int_{AB}\vec{v}_1 \nonumber\\
&=S(AB)\,,
\end{align}
where in the second equality we used Lemma \ref{vAC0}.\footnote{The theorem can also be proven without using Lemma \ref{vAC0}. A similar calculation to \eqref{ABlock}, but without assuming that $\vv_{AC}^+$ and $\vv_{BD}^-$ vanish, shows that
\begin{equation}\label{altproof}
\int_{AB}\vv_1'+\int_{BC}\vv_2' = S(AB)+S(BC)\,.
\end{equation}
Combined with the inequalities $\int_{AB}\vv_1'\le S(AB)$, $\int_{BC}\vv_2' \le S(BC)$, \eqref{altproof} shows that both inequalities are saturated, hence $AB$ and $BC$ are locked.} Thus $\vv_1'$ locks $AB$. A similar calculation shows that $\vv_2'$ locks $BC$.
\end{proof}

It is possible to expand on Theorem \ref{continuumcrossing} to allow for more than two composite regions under certain conditions. Given sets $\I_1,\I_2\subset2^{[n]}$, each of which  can be locked by a single flow (not multiflow) $\vv_1$, $\vv_2$ respectively, there exists a $\nu_a$-multiflow that locks $\I_1\cup\I_2$. A set of regions can be locked by a single flow if they form a nested sequence, $I_1\subset I_2\subset\ldots$ \cite{Headrick:2017ucz}. 
%Matt: added following sentence
For example, with $n=4$, we could choose $\I_1=\{A,AB,ABD\}$ and $\I_2=\{B,BC,ABC\}$, allowing us to lock all elementary regions in addition to $AB$ and $BC$. (Note that nesting is a more restrictive requirement than no-crossing. For example, $A$, $B$, and $AB$ do not cross, but they are not nested, and they may not be lockable by a single flow.) This generalization is an easy extrapolation from Theorem \ref{continuumcrossing}, and its proof is left as an exercise to the reader.

\subsection{Conjectures and open questions}
\label{sec:open questions}

The theorem and extension proved above using the $\nu_a$ density bound use two ``layers'' of flows to lock two sets of regions. With two layers, it is always possible to saturate any two directions using the exact sum and difference technique used above. In fact, this technique may be more powerful than Theorem \ref{continuumcrossing} lets on.

One stronger conjecture, which we call the ``bipartite locking conjecture,'' postulates that the union of two \emph{cross-free} sets $\mathcal{I}_1$, $\mathcal{I}_2$ can be locked with a $\nu_a$-multiflow. This is stronger than Theorem \ref{continuumcrossing} in that each set need not be lockable by a single flow. However, since it is cross-free, each set can be locked with a $\nu_c$-multiflow, which consists of a single layer. Overlaying the two $\nu_c$-multiflows and performing the same sum and difference procedure produces a set of curves that is locally a multiflow and saturates all minimal surfaces. The challenge then comes with choosing the right signs. A $\nu_c$-multiflow may have sharp discontinuities and abruptly switch directions on the boundaries of component flow supports. These discontinuities cause the $\vec{v}_+$ and $\vec{v}_-$ layers to swap places at these boundaries, resulting in some strange thread behavior. A single curve in one layer may, at various places in the manifold, contribute to several different component flows, yet never reach the boundary. Tracing out the threads through these tangling and interchanging layers has proven to be a substantial task. In particular no portion of thread gets reused when converting the two-layer tangle into a proper multiflow. Even if proven, however, the bipartite locking conjecture is insufficient to prove any entropy inequalities beyond MMI. The next simplest inequality, the dihedral inequality \eqref{dihedral}, has the set of regions
\begin{equation}\label{dihedralI}
\I=\{AB,BC,CD,DE,EA,ABCDE\}
\end{equation}
on its right-hand side. The set $\I$ is not bipartite, i.e.\ it cannot be divided into two cross-free subsets. Therefore, to prove the dihedral inequality with a $\nu_a$-multiflow, we would need something stronger.

A natural conjecture, modelled on the standard network locking theorem \cite{LockingTheorem} (given as theorem \ref{NetworkLocking} in the appendix), is that any set of regions that does not contain a pairwise crossing triple can be locked by a $\nu_a$-multiflow. This condition does hold for \eqref{dihedralI}, so this conjecture does imply the dihedral inequality. While this conjecture is very reasonable, it is worth noting that the proof of the analogous network theorem is quite involved, requiring fairly sophisticated combinatorial and network-manipulation techniques. We would not expect the manifold proof to be any simpler. Therefore, even if the conjecture is correct, its proof will almost certainly require techniques that go well beyond those that we have developed in this paper. For inequalities beyond the dihedral one, an even more relaxed condition than the $\nu_a$ bound---perhaps involving layers of threads subject to a joint density bound involving only certain combinations of layers---would seem to be required.

\begin{figure}
\centering
\includegraphics[width=3.5in]{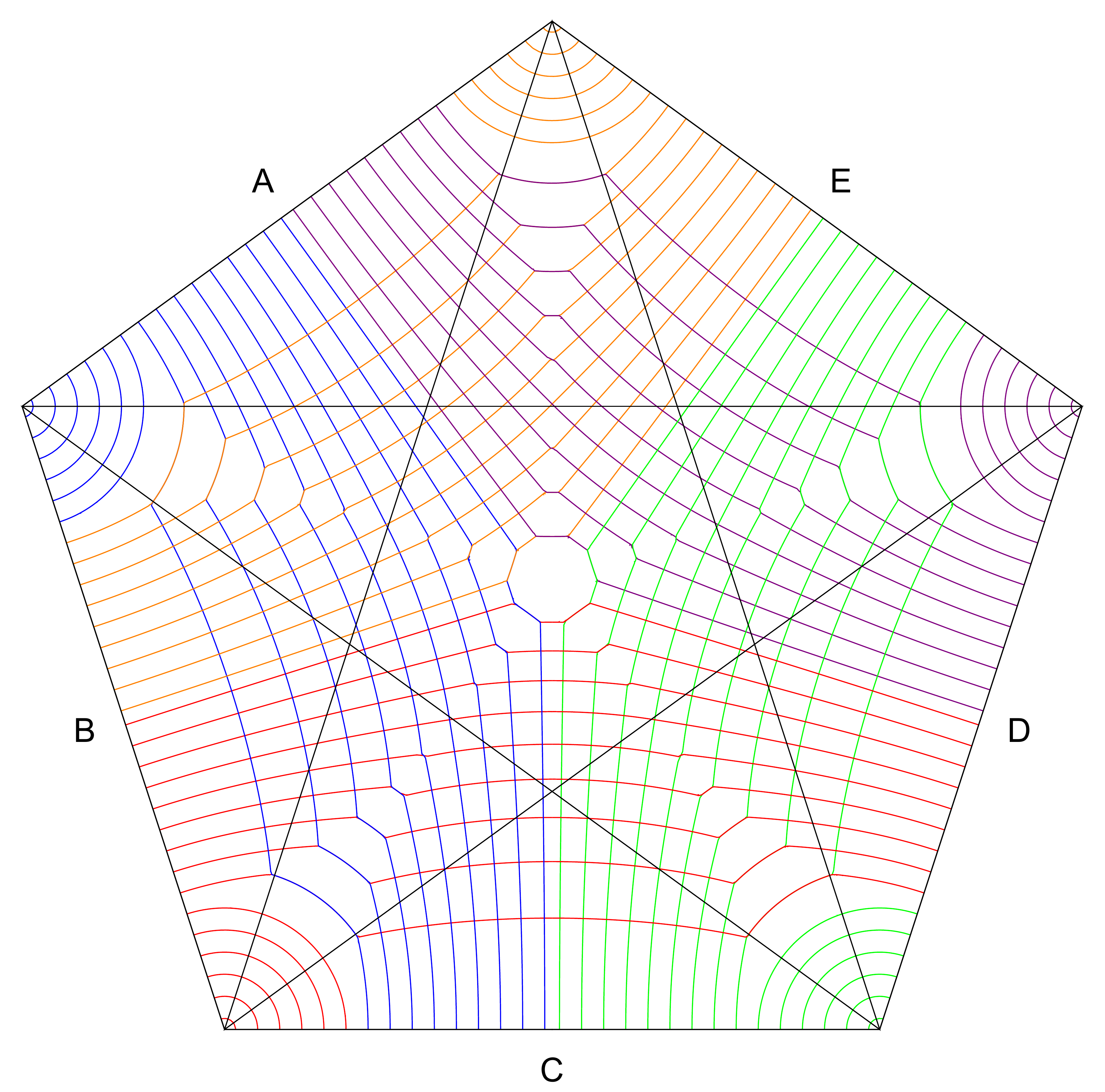}
\caption{
Illustration of $\nu_a$-multiflow on a flat regular pentagon with sides $A,B,C,D,E$ that locks $\mathcal{I}=\{AB,BC,CD,DE,EA\}$. The minimal surfaces are shown as grey chords. This is a special case of the regions appearing on the right-hand side of the dihedral entropy inequality \eqref{dihedral}, with symmetry in the five primitive regions, and empty purifying region. The solution has five distinct layers, each shown in a different color. At any given place in the pentagon, at most two layers overlap at a time ($\zeta\leq 2$).}
\label{fig:Pentagon}
\end{figure}

Non-bipartite sets of regions cannot be locked with two flow layers. However, with more than two layers, restricted to overlapping only two at a time, we can keep all the desirable properties of two local flow layers ($\zeta = 2$ in the language of subsection \ref{subsection:retanglingConvergence}) while locking more complex sets of regions. Figure \ref{fig:Pentagon} was constructed by carefully taking advantage of the symmetries to solve one fifth of the pentagon alone. The solution was then rotated into an array to cover the entire manifold. For problems with less symmetry, constructing such a solution becomes far more challenging, and may not even be possible. However, the success of this simple case offers optimism about the scope of what a $\nu_a$-multiflow restricted to $\zeta\leq 2$ can accomplish.

\begin{table}
\centering
{%
\renewcommand{\arraystretch}{1.2}%
\newcommand{\rc}[2]{\makebox[\linewidth][c]{#1}\makebox[3pt]{}\llap{\textsuperscript{#2}}}%
\newcommand{\gC}{\inlinegraphics{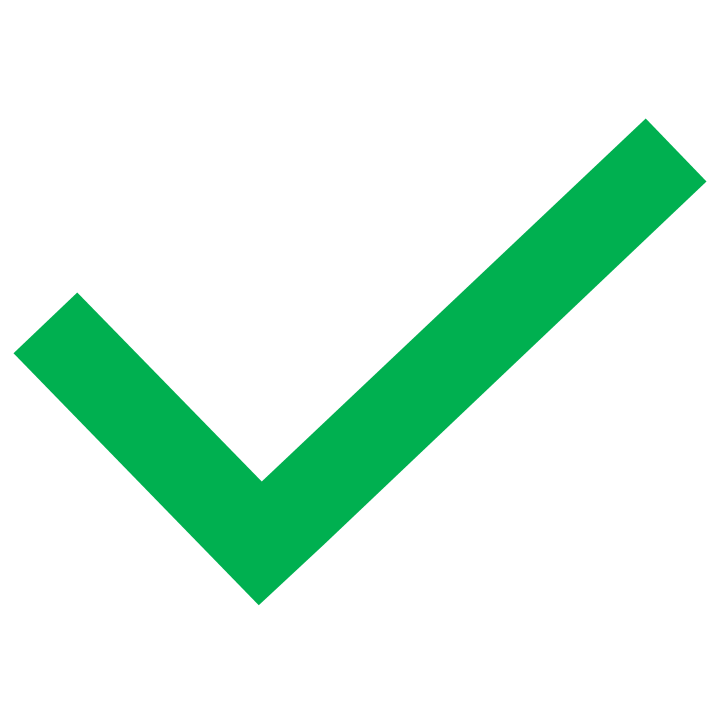}}%
\newcommand{\rX}{\inlinegraphics{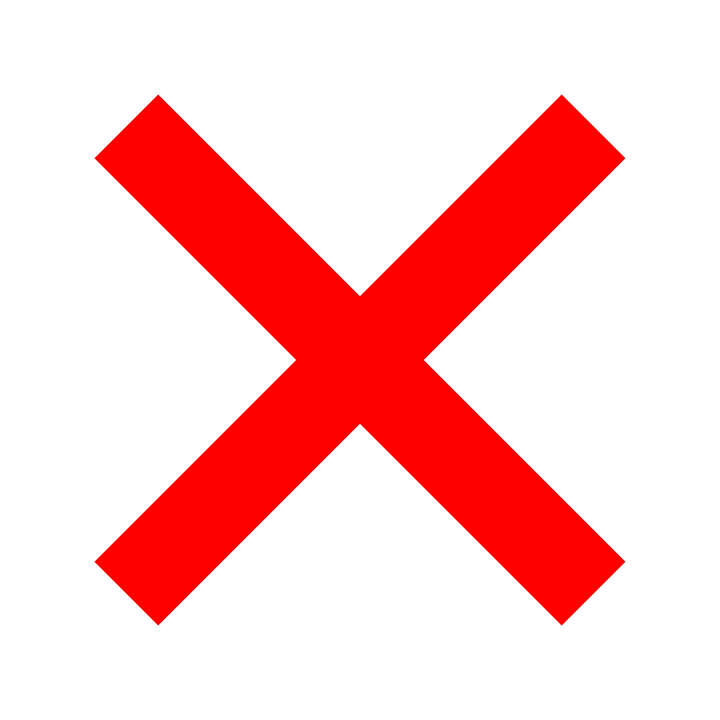}}%
\newcommand{\yQ}{\inlinegraphics{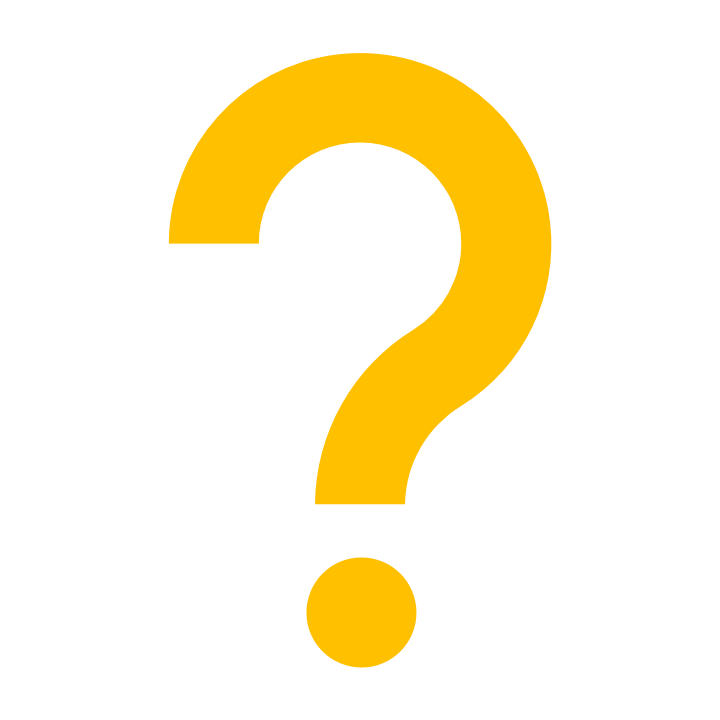}}%
\begin{tabularx}{\textwidth}{V{2.5}XV{2.5}C{0.08\textwidth}|C{0.08\textwidth}|C{0.08\textwidth}|C{0.08\textwidth}|C{0.105\textwidth}V{2.5}}
\hlineB{2.5}
\multirow{2}{=}{\centering \textbf{Set of Regions to Lock}} % header row
	& \multirow{2}{=}{\centering \textbf{Flow}}
	& \multicolumn{3}{c|}{\footnotesize{\textbf{Multiflow}}}
	& \multirow{2}{=}{\centering \footnotesize{\textbf{Network Multiflow}}}\\
	\cline{3-5}
	\renewcommand{\arraystretch}{1.2}
	&& $\boldmath{\nu_c}$	& $\boldmath{\nu_v}$	& $\boldmath{\nu_a}$	&\\
	\hlineB{2.5}
% rest of the table
Nested								& \rc{\gC}{\cite{Headrick:2017ucz}} & \gC & \gC & \gC & \gC \\ \hline
Disjoint							& \rX & \gC & \rc{\gC}{\cite{Cui:2018dyq}} & \gC & \gC \\ \hline
Non-crossing						& \rX & \rc{\gC}{\ref{section:NOLT}} & \rc{\gC}{\ref{section:WCLT}} & \gC & \gC \\ \hline
2 crossing regions 					& \rX & \rX & \rc{\rX}{\ref{section:crossing}} & \rc{\gC}{\ref{section:altBound}} & \gC \\ \hline
2 $\times$ nested (nested bipartite)& \rX & \rX & \rX & \rc{\gC}{\ref{section:altBound}} & \gC \\ \hline
2 $\times$ non-crossing (bipartite)	& \rX & \rX & \rX & \yQ & \gC \\ \hline
3 cross-free						& \rX & \rX & \rX & \yQ & \rc{\gC}{\cite{LockingTheorem}} \\ \hline
3 crossing							& \rX & \rX & \rX & \rX & \rc{\rX}{\cite{LockingTheorem}} \\ \hlineB{2.5}
\end{tabularx}
\begin{tabularx}{\textwidth}{X}
\makebox[\hsize][s]{\mbox{LEGEND:} \mbox{\gC\quad proven} \mbox{\rX\quad disproven} \mbox{\yQ\quad conjectured}}
\end{tabularx}}
\caption{Summary of current knowledge about locking of different sets of regions by different types of multiflows. A green check means that a region set of the indicated type can always be locked by the indicated type of flow or multiflow, while a red ex means that it cannot always be locked (not that it can never be locked). Columns are labeled by flow type, and placed in order of increasing generality. Therefore, if a box is proven (green check), all boxes to the right are also proven. Contrapositively, if a box is disproven (red ex), then all boxes to the left are also disproven. Boxes with references in their top right corners link to where they were first proven (paper citations with square brackets, sections of this paper without). Green checks and red exes without references are implied by the proven boxes.}
\label{summaryTable}
\end{table}

In summary, this paper has introduced two new multiflow constructions (the $\nu_c$- and $\nu_a$-multiflows), and proven several new locking theorems concerning all types of continuum locking theorems. However, there is likely much left to prove. Table \ref{summaryTable} presents a summary of which locking theorems are known, and which are conjectured. It should be noted that any composite boundary regions that can be locked in a continuum multiflow can also be locked on a network, and any composite boundary regions sets that cannot be locked on the network cannot be locked on the continuum. This relationship can be proven by the processes converting a continuum multiflow to a network multiflow, and visa versa. A detailed description of these processes can be found in the appendix.

\acknowledgments{
The work of Headrick and Herman is supported in part by the Simons Foundation through \emph{It from Qubit: Simons Collaboration on Quantum Fields, Gravity, and Information} and in part by the Department of Energy, Office of High-Energy Physics, through Award DE-SC0009987. We would like to thank Veronika Hubeny for very useful conversations and comments on the draft, and John Wilmes for help with the numerical optimization. Headrick would also like to thank the Kavli Institute for Theoretical Physics for hospitality while this work was completed; KITP is supported in part by the National Science Foundation under Grant No.\ NSF PHY-1748958.
}

\appendix

\section{Multiflows on networks}
\label{sec:networks}

Much of the work in this paper is inspired by the well-established theory of multiflows on networks and the question of to what extent that theory carries over to multiflows on manifolds. In this appendix we discuss some aspects of multiflows on networks and their relation to multiflows on manifolds. We begin in subsection \ref{sec:basic} by reviewing basic definitions of networks and path sets, as well as the standard network locking theorem, Theorem \ref{NetworkLocking}, which says that a three-cross-free terminal subset can be locked. In subsection \ref{sec:networkfailure}, we explain the importance of the three-cross-free condition by proving what is essentially a converse result. Theorem \ref{NetworkLocking} involves integral networks and integral paths sets, and its proof is correspondingly discrete and combinatorial in nature. In subsection \ref{sec:networkflows}, we give an alternative description of multiflows in terms of vector fields, analogous to multiflows on manifolds, and in subsection \ref{sec:networkcrossing} we use that desciption to prove Theorem \ref{networkcrossing}, which is a network analogue of Theorem \ref{continuumcrossing}. While this is essentially a special case of Theorem \ref{NetworkLocking}, as far as we know the method of proof is new. Finally, in subsection \ref{sec:networkmanifold}, we explain how to map a manifold to a network and vice versa, and correspondingly map multiflows from one to the other.

More detailed information on multiflows on networks can be found in \cite{MR2676135} and references therein.

\subsection{Basic definitions and theorems}
\label{sec:basic}

A network $\N$ consists of a graph $(V,E)$ with vertex set $V$ and edge set $E$, a subset $T\subseteq V$ of its vertices called \emph{terminals}, and perhaps some extra data such as a positive real- or integer-valued capacity (or weight) $c(e)$ associated to each edge $e\in E$. Various types of graphs are considered in the literature: directed or undirected graphs; simple graphs (with at most one edge connecting any pair of distinct vertices and no edge connecting a vertex to itself) or multigraphs (with any number of edges connecting a given pair of distinct vertices or a vertex to itself); finite or infinite graphs; etc. To focus the discussion, we will consider only finite, undirected, simple graphs, with real or integral capacities. A network will be assumed to have real capacities unless specified as an ``integral network''.\footnote{An integral network $\N$ is essentially equivalent to a network $\N'$ based on a multigraph with unit capacities, where for every edge $e$ with capacity $c(e)$ in $\N$ there are $c(e)$ edges connecting the same pair of vertices in $\N'$. An integral path set on $\N$ is then equivalent to a pairwise edge-disjoint path set on $\N'$. (See below for the definition of a path set.) This equivalent language is commonly used in the network flow literature.}

We will use $x,y$ to denote general vertices and $i,j$ to denote terminals, and take $T=[n]:=\{1,\ldots,n\}$. Given a subset $I\subseteq T$ of terminals, we write $\bar I:=T\setminus I$. An $I$-cut is a set of edges $C\subseteq E$ such that in the graph $(V,E\setminus C)$ there is no path connecting $I$ and $\bar I$. In the analogy between networks and manifolds, the terminal set is the analogue of the boundary of the manifold, each terminal is the analogue of an elementary region $A_i$, and an $I$-cut is the analogue of a bulk surface homologous to $A_I$. The min cut function $S(I)$ is the minimum total capacity of any $I$-cut:
\begin{equation}
S(I):=\min_{\text{$I$-cut $C$}}\sum_{e\in C}c(e)\,.
\end{equation}

The analogue of a ``thread configuration'' is a set $\PP$ of (unoriented simple) paths, each of which connects distinct terminals, together with a positive weight $w_P$ for each path $P\in\PP$, satisfying the capacity constraint:
\begin{equation}\label{capacityconstraint}
\sum_{P\in\mathcal{P},P\ni e}w_P\le c(e)\quad\text{for all }e\in E\,.
\end{equation}
Such a weighted path set is sometimes called a ``multiflow'' in the literature on network flows, but we will reserve the term for a set of vector fields, defined in subsection \ref{sec:networkflows} and analogous to a multiflow on a manifold used in the main text. We will use the term \emph{path set} to mean a weighted path set obeying \eqref{capacityconstraint}, and \emph{integral path set} for one with integer weights.

Given a path set $\PP$, we will denote the total weight of paths connecting disjoint terminal sets $I$, $J$ by $N_{I:J}$:
\begin{equation}
N_{I:J}:=\sum_{\PP\ni P:I-J}w_P\,.
\end{equation}
In particular, the number connecting $I$ and $\bar I$ is bounded above by $S(I)$:
\begin{equation}
N_{I:\bar I}\le S(I)\,.
\end{equation}
The network max flow-min cut theorem says that this inequality is tight:
\begin{equation}\label{networkMFMC}
\max_\PP N_{I:\bar I}= S(I)\,.
\end{equation}
We say that $\PP$ \emph{locks} $I$ if it achieves the maximum. Menger's theorem says further that, on an integral network, $I$ can be locked by an integral path set.

Given a family $\I\subseteq2^T$ of terminal sets, we say that $\PP$ \emph{locks} $\I$ if it locks $I$ for all $I\in\I$. The network locking theorem of Karzanov-Lamonosov gives sufficient conditions for $\I$ to be lockable by an integral path set.
\begin{theorem}[Network locking \cite{KarzanovLomonosov,LockingTheorem}]\label{NetworkLocking}
Let $\N$ be an inner-Eulerian network, and $\I$ a family of subsets of $T$ without a pairwise-crossing triple. Then there exists an integral path set that locks $\I$.
\end{theorem}
\noindent An inner-Eulerian network is an integral network such that, for every interior (i.e.\ non-terminal) vertex $x\in V\setminus T$, the total capacity of the incident edges is even. To see why this condition is needed, consider the simplest non-trivial integral network that is not inner-Eulerian, namely the star graph with 3 external vertices $A,B,C$ and unit capacities. Locking the set $\I=\{A,B,C\}$ requires paths of weight $1/2$ connecting every pair of external vertices. A refined proof of theorem \ref{NetworkLocking} is given in \cite{LockingTheorem}, and will not be repeated here.

The following corollaries follow simply by multiplying all the capacities of $\N$ by 2 (for Corollary \ref{notinnerE}) or twice their lowest common denominator (for Corollary \ref{rationalcapacities}), which makes the network inner-Eulerian, and then dividing the weights in the resulting locking path set by the same factor.

\begin{corollary}\label{notinnerE}
Let $\N$ be an integral network, and $\I$ a family of subsets of $T$ without a pairwise-crossing triple. Then there exists a path set with half-integer weights that locks $\I$.
\end{corollary}

\noindent (The above star graph is an example of Corollary \ref{notinnerE}.)

\begin{corollary}\label{rationalcapacities}
Let $\N$ be a network with rational capacities, and $\I$ a family of subsets of $T$ without a pairwise-crossing triple. Then there exists a path set with rational weights that locks $\I$.
\end{corollary}

\noindent It seems reasonable to expect, by approximating the capacities more and more closely by rationals, that Corollary \ref{rationalcapacities} holds with ``rational'' replaced by ``real''; however, we have neither seen this statement in the literature nor proved it.

\subsection{Network locking failure}
\label{sec:networkfailure}

The three-cross-free condition in Theorem \ref{NetworkLocking} and its corollaries is in fact necessary if one wants lockability of $\I$ to be guaranteed without putting further restrictions on the network, as shown by the following theorem:\footnote{There are locking theorems with additional conditions on the network interior that are able to surpass the three-cross-free ``ceiling''. One example, mentioned in \cite{LockingTheorem}, states that if the network is planar with all its terminals on the exterior, and terminal sets in $\I$ consist only of contiguous spans of the network perimeter, then $\I$ can be locked.}

\begin{theorem}[Network locking failure]\label{lockingfailure}
Let $n$ be an integer $\ge4$, and $\I\supseteq\{I,J,K\}$ a family of subsets of $[n]$ such that $I$, $J$, $K$ are pairwise crossing. Then there exists an inner-Eulerian integral network $\N$ with terminal set $T=[n]$ on which $\I$ cannot be locked by any path set.
\end{theorem}

\noindent (Note that we do not require the path set to be integral, so the theorem is a bit stronger than the converse to Theorem \ref{NetworkLocking}.) Theorem \ref{lockingfailure} was proven in \cite{KarzanovPevzner}, which is in Russian. We are not aware of a proof in English, so we provide one here.

\begin{proof}
We will classify the different types of three-crossing, and for each type find a network on which $\{I,J,K\}$ cannot be locked.

\begin{figure}
\centering
\includegraphics[width=2in]{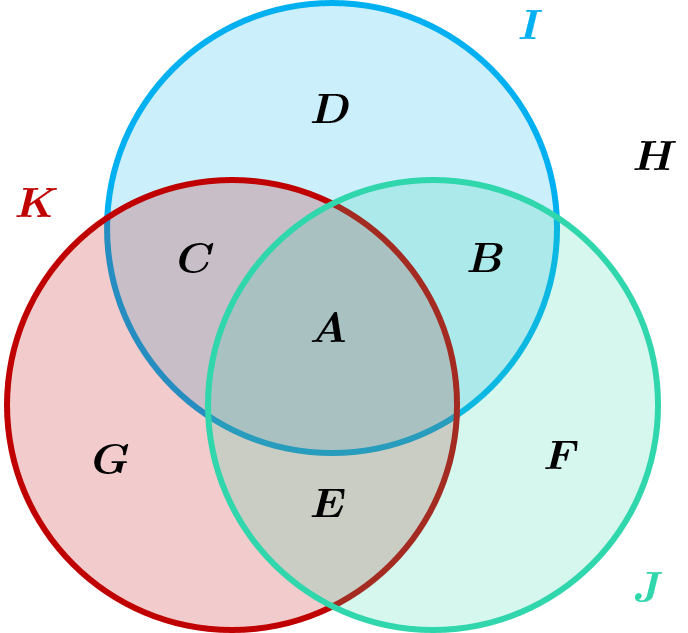}
\caption{Visual representation of the regions defined in \eqref{eq:IntersectionRegions}.}
\label{fig:VenDiagram}
\end{figure}

We decompose $[n]$ into the following eight subsets (see figure \ref{fig:VenDiagram}):
\begin{equation}
\begin{matrix}\label{eq:IntersectionRegions}
&A=I		\cap J			\cap K			&\qquad		E=\bar{I}	\cap J			\cap K			\\
&B=I		\cap J			\cap \bar{K}	&\qquad		F=\bar{I}	\cap J			\cap \bar{K}	\\
&C=I		\cap \bar{J}	\cap K			&\qquad		G=\bar{I}	\cap \bar{J}	\cap K			\\
&D=I		\cap \bar{J}	\cap \bar{K}	&\qquad		H=\bar{I}	\cap \bar{J}	\cap \bar{K}	\\
\end{matrix}
\end{equation}
$H$ is the subset of terminals not in $I$, $J$, or $K$, and will be referred to as the purification set. 
In order for $I$, $J$, $K$ to cross pairwise, we need all of the following sets to be non-empty:
\begin{equation}
\underbrace{AB,EF,CD,GH}_{I\text{ crosses }J},\underbrace{BF,CG,AE,DH}_{J\text{ crosses }K},\underbrace{BD,EG,AC,FH}_{I\text{ crosses }K}
\end{equation}
To avoid double-counting equivalent cases under complementation, we will assume that $H$ is non-empty. This automatically satisfies the requirement that $GH$, $DH$, and $FH$ are non-empty, and we are left with the other nine. We can visually represent these constraints with the following graph, in which the nodes are labeled $A,\ldots,G$ and two terminal sets $X,Y\in\{A,\ldots, G\}$ are adjacent if their union $XY$ must be non-empty:

\noindent%
\begin{minipage}[c]{\textwidth}
\centering
\includegraphics[width=1.5in]{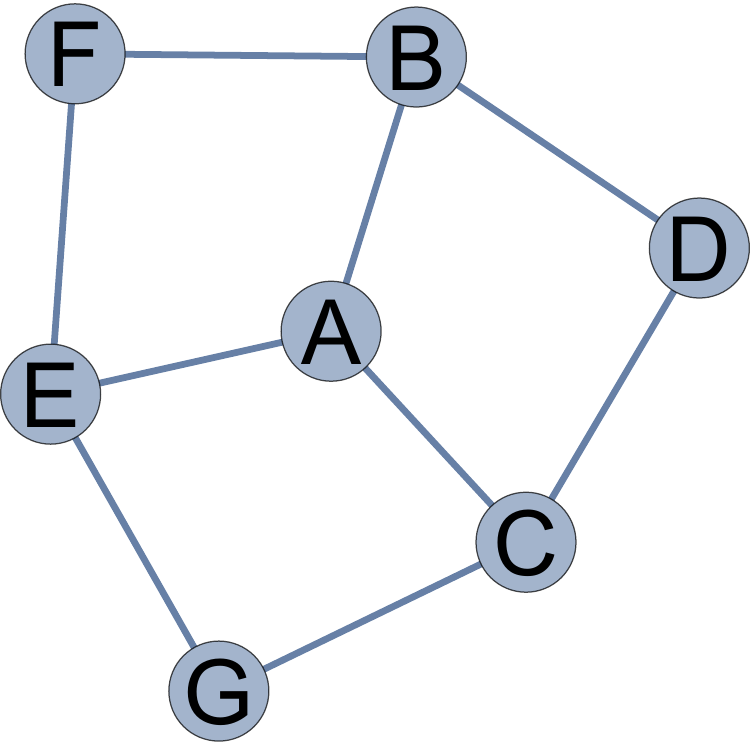}
\end{minipage}
For each edge in the above graph, at least one of the incident vertices must be a non-empty terminal set. There are three distinct minimal cases which satisfy this requirement. Vertices in red are non-empty in their respective cases. All instances of three-crossing are reducible to at least one of these cases:
\medskip

\noindent\begin{minipage}[c]{\textwidth}
\centering
\begin{minipage}[c]{0.25\textwidth}
\centering
\includegraphics[width=1in]{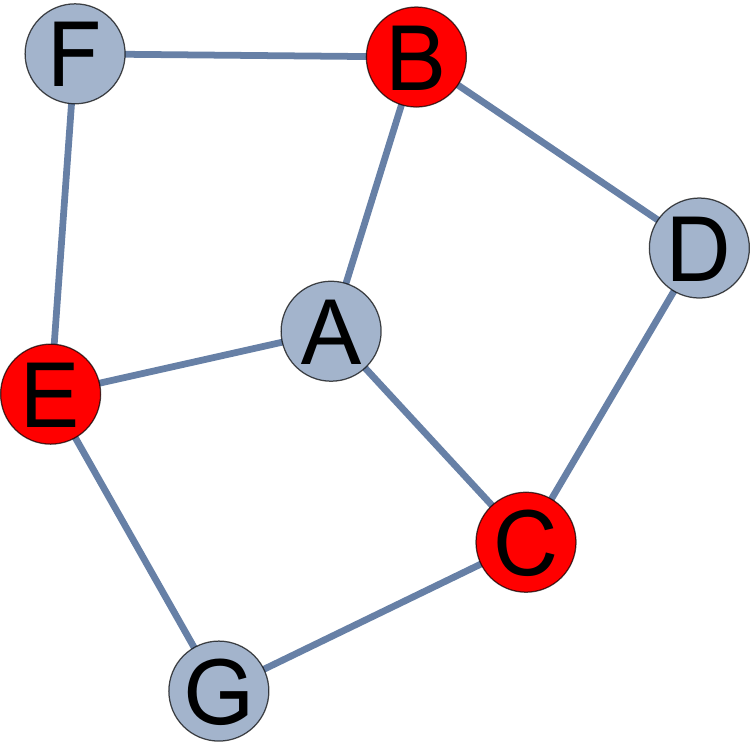}\\
Case 1
\end{minipage}%
\begin{minipage}[c]{0.25\textwidth}
\centering
\includegraphics[width=1in]{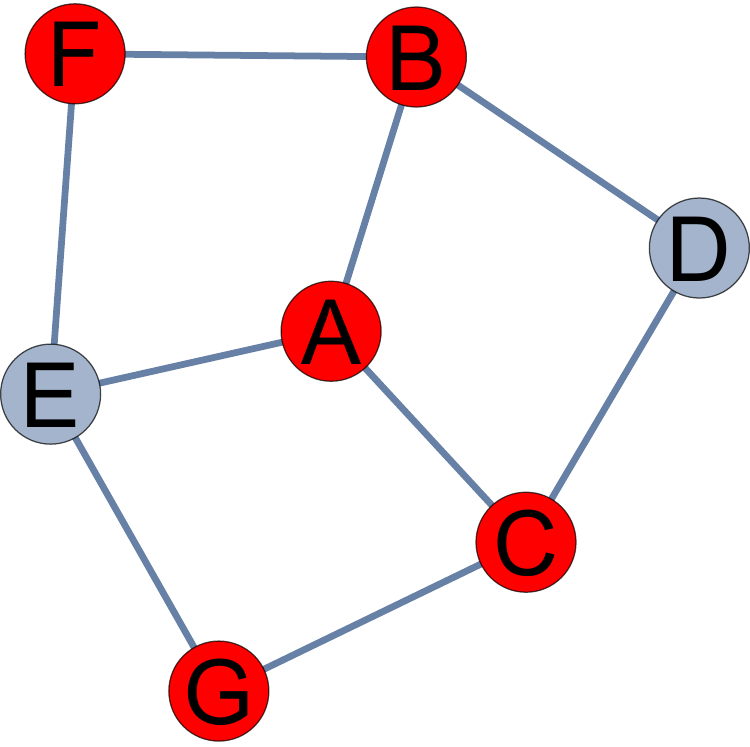}\\
Case 2
\end{minipage}%
\begin{minipage}[c]{0.25\textwidth}
\centering
\includegraphics[width=1in]{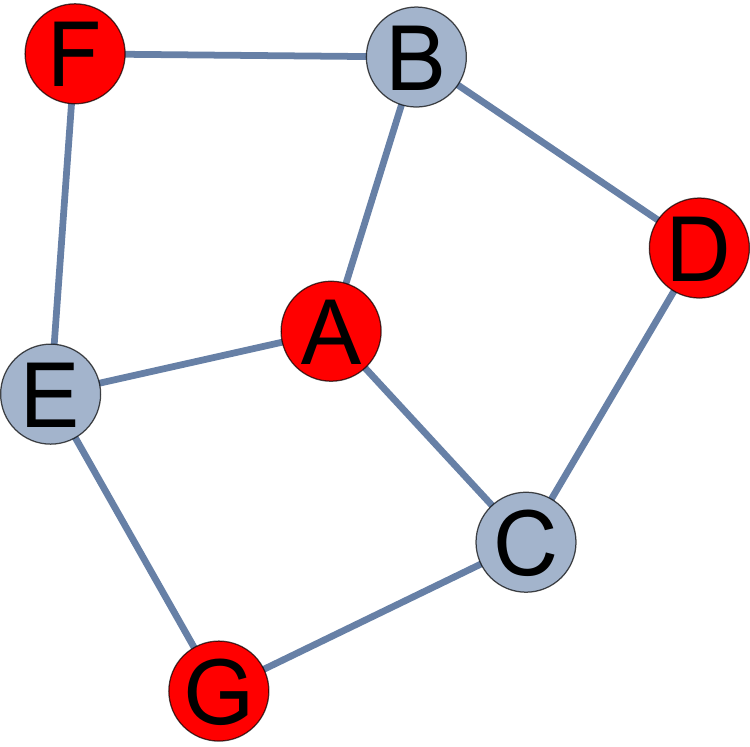}\\
Case 3
\end{minipage}%
\end{minipage}
\smallskip

\noindent As it turns out, cases 1 and 3 are not distinct. In case 3, if we replace $I$, $J$, $K$ with their complements $\bar{I}$, $\bar{J}$, $\bar{K}$ (making $A$ the new purification set), and then set the old purification set $H$ to the empty set, we are left with exactly case 1. This leaves only cases 1 and 2.

In case 1, $B$, $C$, $E$, $H$ are all non-empty; let $b,c,e,h\in[n]$ be elements of the respective subsets. Then $I\supseteq \{b,c\}$, $J\supseteq \{b,e\}$, $K\supseteq\{ c,e\}$. It is a straightforward to show that, on the following graph with terminals $b,c,e,h$ and unit capacities on all edges, the set $\{bc,ce,be\}$ cannot be locked:

\centerline{
{\includegraphics[height=0.75in]{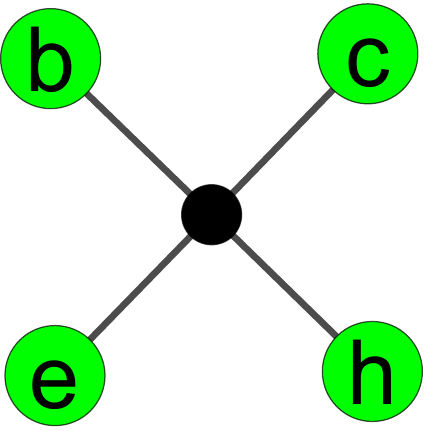}}}

\noindent Adding the elements of $[n]\setminus\{b,c,e,h\}$ as isolated terminals, the result is a network with terminal set $[n]$ in which $\{I,J,K\}$, and therefore $\I$, cannot be locked.

In case 2, $A,B,C,F,G,H$ are non-empty; again let $a,b,c,f,g,h$ be elements. Then $I\supseteq\{ a,b,c\}$, $J\supseteq\{ a,b,f\}$, $K\supseteq\{ a,c,g\}$. 
By the same reasoning, $\I$ is not lockable on the following graph:

\centerline{
{\includegraphics[height=1in]{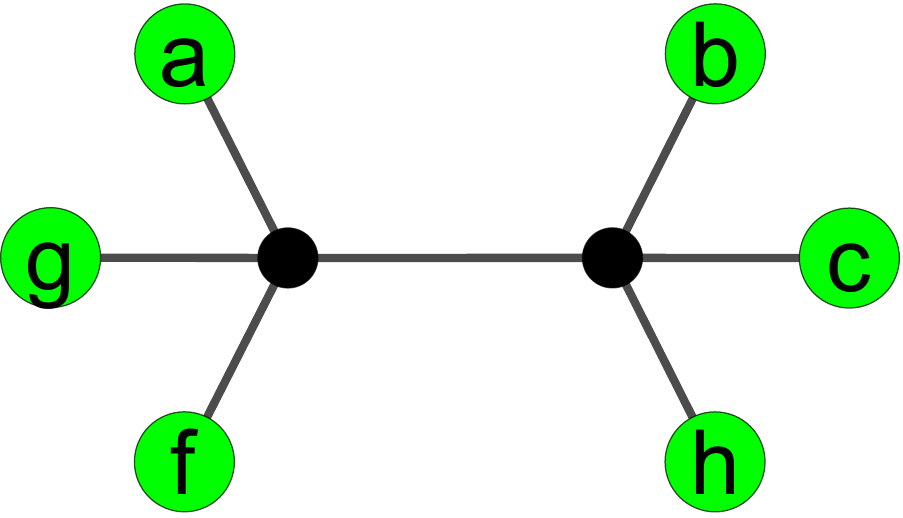}}}
\end{proof}

\subsection{Flows, multiflows, and path sets}
\label{sec:networkflows}

In this subsection, we will define the notions of flows and multiflows on networks in terms of vector fields, by analogy with those concepts on manifolds. We will then show that they can be converted to and from path sets and each other.

First, we associate to each edge $e=\{x,y\}\in E$ a one-dimensional vector space with an inner product---the ``tangent space''. In this tangent space, we label vectors on one side of 0 as being directed from $x$ to $y$ and on the other side from $y$ to $x$. We denote the unit vector from $x$ to $y$ by $\hat n(x,y)$. A \emph{vector field} $\vec v$ associates a vector $\vec v(e)$ to each edge $e$. Given a vector field, the net flux out of a vertex $x$ is defined as
\begin{equation}\label{fluxDef}
\Phi[\vec v](x):=\sum_{y:\{x,y\}\in E}\hat n(x,y)\cdot\vec v(\{x,y\})\,.
\end{equation}
$\Phi$ is a linear map from vector fields to real functions on the vertex set $V$. We also write the total flux out of a set of terminals $I\subseteq T$ as $\Phi[\vec v](I)$:
\begin{equation}\label{uFluxDef}
\Phi[\vec v](I):=\sum_{i\in I}\Phi[\vec v](i)\,.
\end{equation}

A \emph{flow} is a vector field $\vec v$ satisfying the capacity constraint
\begin{equation}
|\vec v(e)|\le c(e)\quad\text{for all }e\in E\,,
\end{equation}
and the divergenceless constraint
\begin{equation}\label{divergenceless}
\Phi[\vec v](x)=0\quad\text{for all }x\in V\setminus T\,.
\end{equation}
The flux of a flow out of a terminal set $I$ cannot exceed its min cut:
\begin{equation}
\Phi[\vec v](I) \le S(I)\,.
\end{equation}
The flow version of the max flow-min cut theorem says that this bound is tight:
\begin{equation}
\max_{\text{flows }\vec v}\Phi[\vec v](I)= S(I)\,.
\end{equation}
This is equivalent to the path-set version of the theorem \eqref{networkMFMC} by converting between flows and path sets, as discussed below. We say $\vec v$ \emph{locks} $I$ if it achieves the maximum.

A \emph{multiflow} is a set of flows $V=(\vec v_{ij})_{i<j}$ obeying the joint capacity condition
\begin{equation}\label{jointcapacity}
\sum_{i<j}|\vec v_{ij}(e)|\le c(e)\quad\text{for all }e\in E\
\end{equation}
and the no-flux condition for $\vec v_{kl}$ on $i$:
\begin{equation}\label{noflux}
\Phi[\vec v_{ij}](k)=0\quad
\text{for all }i,j,k\in T\quad(k\neq i,j)\,.
\end{equation}
Together with the divergenceless condition, \eqref{noflux} implies
\begin{equation}\label{total}
\Phi[\vec v_{ij}](i) = - \Phi[\vec v_{ij}](j)\,.
\end{equation}
Given a multiflow, any constant ($e$-independent) linear combination of the component flows with coefficients between $-1$ and 1 is a flow:
\begin{equation}\label{flowfrommf}
\vec v(e) = \sum_{i<j}\xi_{ij}\vec v_{ij}(e)\,,\qquad-1\le\xi_{ij}\le1\,.
\end{equation}
We say that a multiflow \emph{locks} a set $I$ of terminals if the flow
\begin{equation}
\vec v_I:=\sum_{I\ni i< j\in\bar I}\vec v_{ij}-\sum_{\bar I\ni i< j\in I}\vec v_{ij}
\end{equation}
locks $I$.
As in the continuum, the advantage of multiflows is that they allow us to keep track not just of the total flux out of a given set of terminals but of how much flux is flowing from one terminal to another.

It was indicated in \eqref{flowfrommf} that the component flows of a multiflow can be combined to obtain a single flow. Conversely, as we will show, a flow can be decomposed into a multiflow. We will do this by first writing the flow as a path set and then sorting the paths by their endpoints.

We first define an \emph{(integral) oriented path set} $\tilde\PP$ as a set of paths $P$, each of which has a positive (integer) weight $w_P$ and an orientation $\hat n_P(e)$, which is a unit vector on each edge $e\in P$, obeying the capacity constraint \eqref{capacityconstraint}. We allow oriented path sets to include not only paths connecting distinct terminals but also loops. This is merely a technical convenience, as it makes the correspondence to flows closer and simplifies some formulas. An oriented path set $\tilde\PP$ can be converted to a path set $\PP$ by dropping the orientations and deleting any loops; this preserves any integrality and locking properties of $\tilde\PP$.

An oriented path set $\tilde\PP$ can be converted to a flow $\vv$:
\begin{equation}\label{vfromP}
\vv(e)= \sum_{e\in P\in\tilde\PP}w_P\hat n_P(e)\,.
\end{equation}
Similarly, to write $\tilde\PP$ as a multiflow we assign the paths connecting distinct terminals $i<j$ to the component flow $\vv_{ij}$, and the loops to any component flow.

The converse decomposition is given by the following lemma, which we prove below:
\begin{lemma} \label{pathfromflow}
Let $\N$ be a network and $\vv$ a flow on $\N$. Then there exists an oriented path set $\tilde\PP$ such that \eqref{vfromP} holds and the orientations $\hat n_P(e)$ are equal for all $e\in E$ and $P\in\PP$ containing $e$, implying
\begin{equation}
|\vv(e)| = \sum_{e\in P\in\tilde\PP}|w_P|\,.
\end{equation}
Furthermore, if $\vv$ is integral then there exists such a $\tilde\PP$ that is integral.
\end{lemma}
\noindent By separating the paths according to their endpoints, and converting the paths with endpoints $i<j$ into the component flow $\vv_{ij}$, we obtain a multiflow:
\begin{corollary}
\label{flowdecomp}
Given a flow $\vec v$ on a network $\N$, there exists a multiflow $V$ such that
\begin{equation}\label{goal}
\sum_{i<j}\vec v_{ij} = \vec v
\end{equation}
and, for all $e\in E$, all non-zero $\vec v_{ij}(e)$ have the same direction, implying
\begin{equation}\label{onedirection}
\sum_{i<j}|\vec v_{ij}(e)| = |\vec v(e)|\,.
\end{equation}
\end{corollary}

\begin{proof}[Proof of lemma \ref{pathfromflow}]
We proceed iteratively, building up $\tilde\PP$ step by step. It will be seen that if $\vv$ is integral then the resulting $\tilde\PP$ is also integral.

Let $\vec v^{(0)}:=\vec v$. For $m\ge 0$, if $\vv^{(m)}\neq0$, we define an oriented weighted path $P_m$ and a new flow $\vv^{(m+1)}$. If $\vec v^{(m)}=0$, we stop the procedure, and set $\tilde\PP:=\{P_{m'}\}_{0\le m'<m}$.

Given $\vv^{(m)}$, we find an oriented path $P_m$---either a loop or a path connecting distinct terminals---such that on every edge $e\in P_m$, $\vec v^{(m)}(e)$ is non-zero and is in the same direction as $P_m$. To construct $P_m$, first choose an edge $e_0=\{x_0,x_1\}$ where $\vec v^{(m)}(e_0)\neq0$; assume without loss of generality that $\vec v^{(m)}(e_0)$ is directed from $x_0$ to $x_1$. If $x_1$ is in the interior, then there exists an edge $e_1=\{x_1,x_2\}$ such that $\vec v^{(m)}(e_1)$ is non-zero and directed from $x_1$ to $x_2$. Continue extending the path in this way until it either reaches a terminal or meets a vertex already on the path: $x_{s_1}=x_{s_2}$, $s_1<s_2$. In the latter case let $P_m$ be the loop from $s_1$ to $s_2$; drop the part of the path up to $s_1$. In the former case, continue the path ``backward'' from $x_0$; again, until it either reaches a terminal or a vertex already on the path. Again, in the latter case retain only the loop. This defines $P_m$.

The weight of $P_m$ is fixed as follows,
\begin{equation}\label{vmindef}
w_{P_m} := \min_{e\in P_m}|\vec v(e)|\,,
\end{equation}
and the new flow $\vv^{(m+1)}$ is defined by reducing the flow along $P_m$ by $w_{P_m}$:
\begin{equation}
\vec v^{(m+1)} = \vec v^{(m)}-\begin{cases} w_{P_m}\hat n_{P_m}(e)\,,\quad &e\in P_m\\0\,,\quad &e\notin P_m\end{cases}\,.
\end{equation}
This is still a flow, since the divergenceless constraint is still obeyed, as is the capacity constraint: $|\vec v^{(m+1)}(e)|\le|\vec v^{(m)}(e)|\le c(e)$ for all $e\in E$. Furthermore, the flow has been zeroed out on the edge $e_{\rm min}$ where the minimum in \eqref{vmindef} is achieved: $\vec v^{(m)}(e_{\rm min})\neq0$ but $\vec v^{(m+1)}(e_{\rm min})=0$. Hence, if we repeat this procedure, given that the graph is finite, it will end with the vanishing flow in a finite number of steps. Finally, on any edge $e$ where $\vec v^{(m)}(e)=0$, $\vec v^{(m+1)}(e)=0$, and on any edge $e$ where $\vec v^{(m)}(e)\neq0$, $\vec v^{(m+1)}(e)$ either vanishes or is in the same direction as $\vec v^{(m)}(e)$. In other words, the flow direction is never reversed by the update. Hence $\hat n_{P_m}(e)$ is the same for all $P_m$ containing $e$.
\end{proof}

\subsection{Crossing-pair locking theorem}
\label{sec:networkcrossing}

In this subsection we use flows and multiflows to prove the following theorem. The proof follows closely that of Theorem \ref{continuumcrossing}.
\begin{theorem}[Network crossing-pair locking]\label{networkcrossing}
Let $\N$ be a network and $I,J\subset T$ terminal sets that cross. Then there exists a path set that locks $\{I,J\}$. Furthermore, if $\N$ is integral then there exists a path set with half-integral weights that locks $\{I,J\}$.
\end{theorem}
\noindent The first part of this theorem is a special case of Corollary \ref{rationalcapacities}, or more precisely its conjectured real version, while the second part is a special case of Corollary \ref{notinnerE}.

\begin{proof}
For clarity we consider the case where there are only four terminals, which we denote $A,B,C,D$, and we wish to lock $\{AB,BC\}$. This is without loss of generality. Given crossing sets of terminals $I,J$ on $\N$, we can make a new network $\N'$ by adding a new terminal $A$  attached with edges of infinite capacity to the terminals in $I\setminus J$, and convert the latter into interior vertices; similarly with $B$ and $I\cap J$; $C$ and $J\setminus I$; and $D$ and $\overline{I\cup J}$. A path set on $\N'$ that locks $\{AB,BC\}$ defines a path set on $\N$ that locks $\{I,J\}$.

We now proceed to the proof. We start with max flows $\vec v_{1,2}$ for $AB,BC$ respectively:
\begin{equation}\label{v12maxflows}
\Phi[\vec v_1](AB) = S(AB)\,,\qquad   \Phi[\vec v_2](BC) = S(BC)\,.
\end{equation}
We would like to combine these into a single multiflow, while preserving the fluxes on $AB$ and $BC$. We could use Corollary \ref{flowdecomp} to decompose each of $\vec v_1$ and $\vec v_2$ into a multiflow. The problem is that these two multiflows are not necessarily compatible with each other; we cannot necessarily combine them while obeying the joint capacity constraint \eqref{jointcapacity}.

Therefore we will use a trick. We combine them into flows $\vec v^\pm$:
\begin{equation}
\vec v^\pm :=\frac12\left(\vec v_1\pm\vec v_2\right).
\end{equation}
Obviously this reshuffling does not lose any information about the original flows. However, unlike the original flows, $\vec v^\pm$ are compatible with each other, i.e.\ they obey a joint capacity constraint:
\begin{equation}\label{jointcc}
|\vec v^+(e)|+|\vec v^-(e)|\le c(e)
\end{equation}
for every edge. \eqref{jointcc} follows from the fact that the left-hand side equals the maximum of $|\vec v_1(e)|$ and $|\vec v_2(e)|$. Since they are compatible, when we convert each into a multiflow, we can combine them into a single multiflow, in a way that preserves the relevant fluxes.

We thus use Corollary \ref{flowdecomp} to convert each flow $\vec v^\pm$ into a multiflow $\{\vec v^\pm_{ij}\}$. Each of these obeys \eqref{onedirection}, that is
\begin{equation}
\sum_{i<j}|\vec v_{ij}^\pm(e)|=|\vec v^\pm(e)|
\end{equation}
for all $e\in E$. Combining this with \eqref{jointcc}, we have
\begin{equation}
\sum_{i<j}\left(|\vec v^+_{ij}(e)|+|\vec v^-_{ij}(e)|\right)\le c(e)\,.
\end{equation}
This implies that an arbitrary linear combination of $\{\vec v^+_{ij}\}$ and $\{\vec v^-_{ij}\}$ of the form
\begin{equation}
\vec v_{ij} = \xi^+_{ij}\vec v^+_{ij}+\xi^-_{ij}\vec v^-_{ij}\,,
\end{equation}
with ($e$-independent) coefficients $\xi^\pm_{ij}\in[-1,1]$, obeys the joint capacity constraint, as well as the other conditions in the definition of a multiflow. We choose the following combinations:
\begin{align}\label{newmf}
\vec v_{AB} &= -\vec v^+_{AB}+\vec v^-_{AB} \\
\vec v_{AC} &= \vec v^-_{AC} \nonumber \\
\vec v_{AD} &= \vec v^+_{AD}+\vec v^-_{AD} \nonumber \\
\vec v_{BC} &= \vec v^+_{BC}+\vec v^-_{BC}\nonumber  \\
\vec v_{BD} &= \vec v^+_{BD} \nonumber \\
\vec v_{CD} &= \vec v^+_{CD}-\vec v^-_{CD}  \,.\nonumber
\end{align}

We now show that the multiflow \eqref{newmf} locks $AB$ and $BC$. From \eqref{newmf} we define the following flows:
\begin{equation}
\vec v_1' = \vec v_{AC}+\vec v_{AD}+\vec v_{BC}+\vec v_{BD}\,,\qquad
\vec v_2' = -\vec v_{AB}-\vec v_{AC}+\vec v_{BD}+\vec v_{CD}\,.
\end{equation}
We now calculate the total flux of $\vec v_1'$ out of $AB$:
\begin{align}\label{ABflux}
\Phi[\vec v_1'](AB)&=\Phi[\vec v_{AC}+\vec v_{AD}](A)+
\Phi[\vec v_{BC}+\vec v_{BD}](B)\nonumber\\
&=\Phi[\vv^++\vv^-](AB)-\Phi[\vv^+_{AC}](A)-\Phi[\vv^-_{BD}](B)\nonumber\\
&=S(AB)-\Phi[\vv^+_{AC}](A)-\Phi[\vv^-_{BD}](B)\,.
\end{align}
A similar calculation shows that
\begin{equation}\label{BCflux}
\Phi[v_2'](BC) = S(BC) +\Phi[\vv^+_{AC}](A)+\Phi[\vv^-_{BD}](B)\,.
\end{equation}
Finally, given that
\begin{equation}
\Phi[\vec v_1'](AB)\le S(AB)\,,\qquad
\Phi[\vec v_2'](BC)\le S(BC)\,,
\end{equation}
\eqref{ABflux}, \eqref{BCflux} imply that both bounds must be saturated, i.e.\ $\vec v_1'$ locks $AB$ and $\vec v_2'$ locks $BC$. Each component of the multiflow \eqref{newmf} can be converted using Lemma \ref{pathfromflow} into a path set, the union of which locks $AB$, $BC$.

If $\N$ is integral, then by Menger's theorem $\vv_{1,2}$ can be chosen to be integral. $\vv^\pm$ are then half-integral, and therefore so is the multiflow \eqref{newmf} and the resulting path set.
\end{proof}

\subsection{Relation between network and manifold multiflows}
\label{sec:networkmanifold}

\begin{figure}
    \centering
    \includegraphics[width=4in]{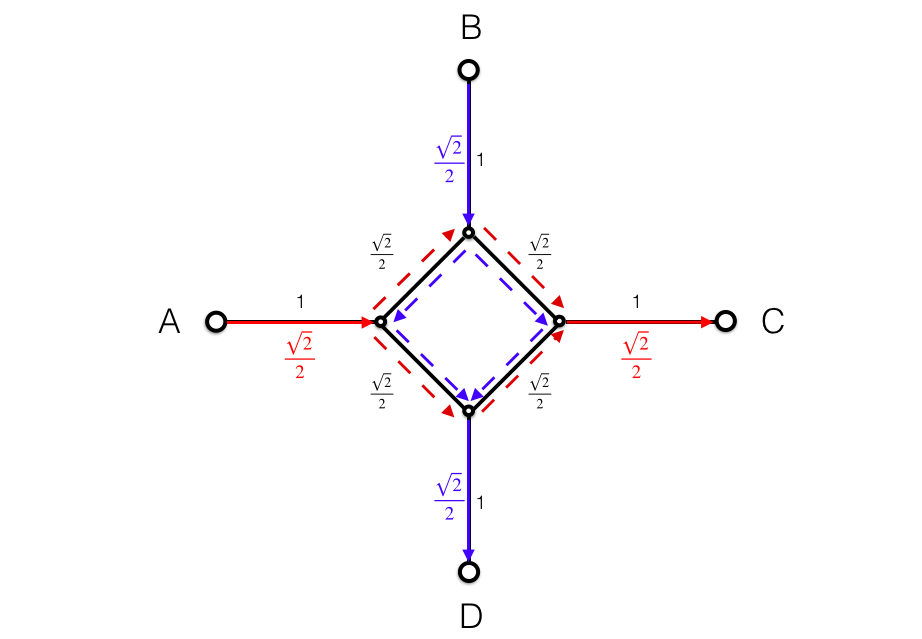}
    \caption{Network obtained by dessication of flat unit square with sides $A,B,C,D$. (More precisely, the square with padding added to separate each boundary region from its respective minimal surface; e.g.\ a disk of radius $1/\sqrt2$ with boundary circle divided into equal arcs.) The capacity of each edge is shown in black. Also shown (in color) is a multiflow that locks both $AB$ and $BC$. The multiflow has only two non-zero components, $\vec v_{AC}$ (red) and $\vec v_{BD}$ (blue).
    }
    \label{fig:DiscreteSquare}
\end{figure}

Multiflows on networks and on manifolds are closely related, since either one can be converted into the other. Given a Riemannian manifold with boundary $\M$ and boundary decomposition $\{A_1,\ldots,A_n\}$, we can define a network $\N$ with terminal set $[n]$ such that the respective min cut functions $S_\M$, $S_{\N}$ agree:
\begin{equation}\label{Sagree}
S_{\M}(A_I)=S_{\N}(I)\quad\text{for all }I\subseteq[n]\,.
\end{equation}
The construction, described in \cite{Bao:2015bfa} and sometimes called ``dessication'', proceeds by cutting $\M$ up along the minimal surfaces $m(A_I)$ for all possible composite regions $A_I$. Each resulting bulk region, or cell, is mapped to a vertex of the network, which is labelled a terminal if the region is bounded by an $A_i$. If two cells are adjacent then the corresponding vertices are joined by an edge, with capacity given by the area of their shared boundary. Under this construction, the minimal cut in $\N$ for any $I\subseteq[n]$ lifts to the minimal surface $m(A_I)$ in $\M$. The dessication of the flat unit square studied in section \ref{section:crossing} is shown in figure \ref{fig:DiscreteSquare}.

A flow $\vv_\M$ on $\M$ can be mapped to a flow $\vv_\N$ on the dessication of $\M$ by setting $\vv_\N(e)$ on a given edge $e$ to the flux of $\vv_\M$ over the corresponding surface in $\M$. This guarantees that, for any terminal set $I\subseteq [n]$,
\begin{equation}
\Phi[\vv_\N](I) = \int_{A_I}\vv_\M\,;
\end{equation}
in particular, $\vv_\N$ locks $I$ if and only if $\vv_\M$ locks $A_I$. Similarly, a $\nu_a$-, $\nu_v$-, or $\nu_c$-multiflow on $\M$ can be mapped to a multiflow on $\N$, preserving the fluxes of all component flows. Figure \ref{fig:DiscreteSquare} shows the multiflow obtained in this way from the $\nu_a$-multiflow on the unit square that locks $AB$ and $BC$, discussed in section \ref{section:altBound} (in which $\vv_{AC}$ is a constant horizontal vector field with norm $1/\sqrt2$ and $\vv_{BD}$ is a constant vertical vector field with the same norm).

In the other direction, a network $\N$ can be ``hydrated'' to form a manifold $\M$ and boundary decomposition $\{A_i\}$ such that \eqref{Sagree} is again satisfied. One such construction is as follows. Each internal vertex of $\N$ is mapped to a very large sphere (of radius much larger than $\max_{e\in E}w(e)$), and each terminal $i\in[n]$ to a very large hemisphere whose boundary is labelled $A_i$. If $\N$ contains an edge $e=\{x,y\}$, then the (hemi)spheres corresponding to vertices $x$, $y$ are sewn together along a circle of circumference $w(e)$, with different circles well separated on each sphere. Again, under this construction the minimal surface $m(A_I)$ in $\M$ corresponds to the minimal cut in $\N$ for the terminal subset $I$.\footnote{We leave the proofs of this and other statements in this subsection to the reader.} Note that dessication and hydration are not inverse operations, in the sense that the dessication of the hydration of a network $\N$ does not necessarily equal $\N$ (a simple counterexample being a graph with one terminal and one internal vertex; after hydration and dessication, the internal vertex is lost).

Under hydration, a flow $\vv_\N$ can be lifted to a flow $\vv_\M$, while preserving fluxes, by setting $\vv_\M$ on the sewing circle corresponding to an edge $e$ equal to $|\vv_\N(e)|$ times the unit normal on the circle (with the obvious orientation); the fact that the spheres are large and the sewing circles well-separated guarantees that $\vv_\M$ can be extended between the circles while being divergenceless and obeying the norm bound. Similarly, a multiflow on $\N$ can be lifted to a $\nu_v$- or $\nu_a$-multiflow on $\M$.\footnote{The component flows may overlap on $\M$ in this construction, so the result is not necessarily a $\nu_c$-multiflow. A $\nu_c$-multiflow can be produced by going to higher dimensions as follows: Map each vertex of $\N$ to a 3-sphere; for each edge $e$ sew the appropriate 3-spheres along a 2-sphere of area $w(e)$; decompose the 2-sphere into non-overlapping regions of area $|\vv_{\N ij}(e)|$ for all $i<j$; set $\vv_{\M ij}$ equal to the unit normal vector on its respective region; and extend the $\vv_{\M ij}$s to the rest of $\M$ to divergenceless vector fields obeying the $\nu_c$ norm bound. The reason it is necessary to go to higher dimensions for this construction is that, on 2-spheres, the component flows (or threads) may be topologically forced to cross each other.} An implication is that Theorem \ref{lockingfailure} (network locking failure) can be lifted directly to manifolds:

\begin{corollary}
Let $n$ be an integer $\ge4$ and $\I\subseteq 2^{[n]}$ a family of terminal sets containing a pairwise crossing triple. Then there exists a manifold $\M$ and boundary decomposition $\{A_1,\ldots,A_n\}$ such that $\I$ cannot be locked by a $\nu_a$-multiflow (and therefore cannot be locked by a $\nu_c$- or $\nu_v$-multiflow).
\end{corollary}

\bibliographystyle{JHEP}
\bibliography{refs}

\end{document}